\documentclass[11pt]{article}

\usepackage[utf8]{inputenc}

\usepackage{fullpage}
\usepackage{amsthm}
\usepackage{dsfont}
\usepackage{amsmath}
\usepackage{amssymb}
\usepackage{upgreek}
\usepackage{esint}
\usepackage{esvect}
\usepackage{graphicx}
\usepackage[margin=1in]{geometry}
\usepackage{bm}
\graphicspath{ {img/} }
\usepackage{mathtools}
\usepackage{mathrsfs}
\usepackage{nicefrac}
\usepackage[ruled,vlined]{algorithm2e}
\usepackage{enumitem}
\usepackage[dvipsnames]{xcolor}

\usepackage{thmtools, thm-restate}
\usepackage{hyperref}
\hypersetup{
    colorlinks=true,
    linkcolor=blue,
    citecolor=violet   
    }
\usepackage[nameinlink]{cleveref}
\usepackage{array}
\usepackage{tabularx}
\usepackage{ragged2e}
\usepackage{makecell}
\usepackage{mathtools}
\newcolumntype{C}[1]{>{\centering\arraybackslash}p{#1}}
\usepackage{authblk}

\usepackage[normalem]{ulem} 

\usepackage{physics}
\usepackage{pgfplots}
\usepackage{algpseudocode}

\newtheorem{theorem}{Theorem}[section]

\newtheorem{proposition}[theorem]{Proposition}

\newtheorem{definition}[theorem]{Definition}

\newtheorem{construction}{Construction}

\newtheorem{example}[theorem]{Example}
\usepackage{tikz}
\usepackage{tikz-cd}
\newcommand{\tabcite}[1]{\\\makebox[\linewidth][r]{#1}}

\newcommand{\Hull}{\mathrm{Hull}}
\newcommand{\C}{\mathbb{C}}
\newcommand{\Z}{\mathbb{Z}}

\newcommand{\F}{\mathbb{F}}

\theoremstyle{definition}

\newtheorem{corollary}[theorem]{Corollary}

\newtheorem{remark}[theorem]{Remark}

\newcommand{\rmv}[1]{}
\renewcommand{\ev}{\mathrm{ev}}

\title{Entanglement-assisted quantum locally recoverable codes: bounds and constructions with availability}
\author[1]{Rutuja Kshirsagar\thanks{Email: rkshirsagar@fujitsu.com}}\affil[1]{Fujitsu Research of America, Inc., Santa Clara, CA, U.S.A}

\author[2]{Gretchen L. Matthews\thanks{G. L. Matthews and the work at Virginia Tech is partially supported by NSF DMS-2502705, NSF CNS-2413218, and the Commonwealth Cyber Initiative. Email:  gmatthews@vt.edu}}\affil[2]{Virginia Tech, Blacksburg, VA, U.S.A.}

\author[3]{Julia Shapiro\thanks{J. Shapiro is also supported by the Department of Defense Cyber Service Academy Scholarship. Email: juliams22@vt.edu}}\affil[2]{Virginia Tech, Blacksburg, VA, U.S.A.}
\date{}

\begin{document}

\maketitle

\begin{abstract}
In this work, we define entanglement-assisted quantum locally recoverable codes with availability, 
in which any set of up to \(\delta-1\) erased qudits can be recovered from any one of \(t\) local recovery sets, each of size at most \(r+\delta-1\), with the recovery sets intersecting exactly in the erased coordinates, where $r$ is a (small) positive integer. We show that shared entanglement permits \(t>1\), meaning that multiple local recovery sets can be available for the same set of up to $\delta -1$ erasures.
We establish a Singleton-like bound for this family of codes and present random constructions based on classical linear codes with Vandermonde parity-check matrices. We also provide explicit constructions of entanglement-assisted quantum locally recoverable codes with availability from several classical code families and their folded versions, including Tamo--Barg codes, fiber-product codes, and algebraic-geometry codes such as one-point Hermitian and Suzuki codes.
\end{abstract}

\section{Introduction}

Locally recoverable codes (LRCs) were introduced in \cite{LRC_intro} for  applications in distributed storage systems. An $(n,k,r)$ LRC over the alphabet $\mathbb{F}_q$ provides an efficient framework for erasure recovery of alphabet symbols, encoding $k$ symbols into $n$ symbols in such a way that each codeword symbol can be recovered from a collection of at most $r$ other symbols, referred to as a recovery set. Since any linear code of dimension $k$ naturally has recovery sets of size $k$ for each symbol, the focus is typically on LRCs with small locality $r<k$. This notion was generalized to $(r,\delta)$ LRCs \cite{KPLK14}, where each recovery set can tolerate up to $\delta-1$ erasures and further extended to include availability $t$ in \cite{availability_intro}, meaning that each symbol admits $t$ (disjoint) recovery sets. Explicit instances of LRCs based on families of algebraic-geometry codes over various curves have been demonstrated in \cite{ballentine.barg.vladut.hlrc,  Barg_2015, Haymaker2016LocallyRC, kadiyam_das_2020_lrc_with_availibility,  tamo.barg.lrc}. We refer to this as the classical setting, in contrast to this paper's focus on quantum error-correcting codes. 

The notion of quantum local recovery was introduced in \cite{golowich2025quantum} and extended to the $(r,\delta)$ setting in \cite{galindo2024quantumrdeltalocallyrecoverablecodes}, giving rise to quantum locally recoverable codes (QLRCs). QLRCs provide a promising framework for future quantum data storage systems \cite{sharma2024quantumlocallyrecoverablecodes}. 
Constructions have been studied in recent works including \cite{galindo2026quantumrdeltabchhomothetic, li2025optimalquantumlrcshermitian,li2025improvedboundspurequantumlrc, luo2025boundsconstructionsquantumcss,zhou2025optimalquantumrdeltalocallyrepairable}. Like their classical counterparts, QLRCs are developed as a step toward achieving quantum distributed storage. In such a model, an erroneous qubit (node) can be recovered by connecting to a small set of local qubits rather than requiring the device to have an all-to-all connectivity, thus minimizing the possibility of long-range connections. Erasure models assume that the positions of lost symbols or qubits are known in advance, reducing decoding to a reconstruction problem rather than error identification. This makes erasures fundamentally easier to handle than general errors in both classical and quantum settings \cite{Cover2006,gottesman1997stabilizercodesquantumerror,Grassl_1997}. Recent experimental results further support this approach, showing that in platforms such as neutral-atom and superconducting systems, detectable faults can be mapped to erasures \cite{Teoh_2023, Wu_2022}. Thus, although QLRCs are motivated by erasure recovery, they support quantum error correction more broadly by enabling low-overhead local repair of detectable faults.
More broadly, QLRCs are expected to be relevant to future quantum memory architectures.

Many leading quantum error-correcting code families are Calderbank–Shor–Steane (CSS) codes, including surface codes \cite{BravyiKitaev1998,  ErrorCorrectionZooSurfaceCode, Kitaev2003} and numerous quantum low-density parity-check (QLDPC) constructions \cite{BreuckmannEberhardt2021}, which require dual containment of the underlying classical codes. 
Entanglement-assisted quantum error-correcting codes (EAQECCs) were introduced in \cite{brun2006correcting} to relax this dual-containment constraint required in CSS code constructions \cite{Calderbank_1996,Steane_1996}. EAQECCs not only allow more flexible constructions, but can also improve the achievable dimension–distance tradeoff by leveraging pre-shared entanglement \cite{brun2006correcting,hsieh2007eaqecc}. These codes directly inherit the structure of classical linear codes, allowing them to integrate with well-established classical coding theory. Moreover, pre-shared entanglement can improve quantum coding performance: entanglement-assisted quantum error-correcting codes may achieve better parameters than comparable unassisted schemes by relaxing orthogonality constraints and enabling constructions from arbitrary classical linear codes~\cite{Grassl2021EA}.

In \cite{guruswami2026quantumhierarchicallocallyrecoverable}, the authors present CSS $(r,\delta)$ QLRCs, including both random and explicit constructions. In \cite{golowich2025quantum}, the authors show that availability is not possible in the case of quantum CSS constructions when using disjoint recovery sets that only intersect at the (erased) coordinate $i$. However, it is possible to have more than one recovery set (correcting a single erasure) if the sets are allowed to overlap in more than one coordinate \cite[Proposition 23]{bu2025quantumlocallyrecoverablecode}. As an illustration, \cite[Example 24]{bu2025quantumlocallyrecoverablecode} demonstrates that constructing the Steane code from two Hamming codes yields a QLRC with $4$ recovery sets of size $3$ for each qudit and the overlap at most $2$. However, availability in the classical sense requires disjoint recovery sets, so that erased symbols in one recovery set do not impact the ability to recover an erasure from another recovery set. That is the motivation for this work: utilizing entanglement allows for multiple disjoint recovery sets for quantum codes.

In this work, we introduce $(r,\delta)$ entanglement-assisted quantum locally recoverable codes (EAQLRCs) with availability $t$, with the property that any set of $\delta-1$ qudits can be recovered from one of $t$ 
subsets of size at most \(r + \delta - 1\), each containing the erased qudits, and whose pairwise intersections are exactly the erased set. 
Such codes combine the advantages of local recovery with the flexibility provided by shared entanglement. {
Consequently, one can build an EAQLRC from any two classical LRCs, including codes that are not dual-containing. This additional flexibility is particularly valuable for LRCs with locality and availability constraints, where enforcing dual containment may significantly restrict the achievable parameters.} We determine a Singleton-like bound for these codes. We provide both random and explicit constructions of EAQLRCs, where the underlying codes are from several families of classical LRCs, including fiber product codes, Tamo–Barg codes, folded Tamo–Barg codes, one-point Hermitian codes, and codes from Suzuki curves.

We summarize the known results for QLRCs alongside our results for EAQLRCs in \Cref{table:QLRC_bounds}. 
\begin{table*}[t]
\centering
\small
\renewcommand{\arraystretch}{1.35}
\setlength{\tabcolsep}{3pt}

\newcolumntype{Y}{>{\RaggedRight\arraybackslash}X}
\newcommand{\cellmath}[1]{%
  $\begin{aligned}#1\end{aligned}$%
}

\begin{tabularx}{\textwidth}{|>{\RaggedRight\arraybackslash}p{0.08\textwidth}|Y|Y|Y|}
\hline
&
\centering $(r,2)$ QLRCs \rmv{\cite{golowich2025quantum}}\arraybackslash
&
\centering $(r,\delta)$ QLRCs \rmv{\cite{guruswami2026quantumhierarchicallocallyrecoverable}}\arraybackslash
&
\centering {$(r,\delta)$ EAQLRCs \rmv{[this work]}}\arraybackslash
\tabularnewline
\hline

Singleton (upper) bound on $k$
&
\makecell[l]{
\scriptsize
\cellmath{
& n-2(d-1) \\
&-\left\lfloor \frac{n-(d-1)}{r+1} \right\rfloor \\
&-\left\lfloor
\frac{
n-2(d-1)-\left\lfloor \frac{n-(d-1)}{r+1} \right\rfloor
}{r+1}
\right\rfloor
}
\\
\tabcite{\cite[Theorem 35]{golowich2025quantum}}
}
&
\makecell[l]{
\scriptsize
\cellmath{
& n-2(d-1) \\
&-\left\lfloor
\frac{n-(d-1)}{(r+\delta-1)}
\right\rfloor \\
&-\left\lfloor
\frac{
n-2(d-1)
-\left\lfloor \frac{n-(d-1)}{r+\delta-1} \right\rfloor
}{r+\delta-1}
\right\rfloor
}
\\
\tabcite{[Extension of \cite[Theorem 35]{golowich2025quantum}]}
}
&
\makecell[l]{
\scriptsize
\cellmath{
& n-2(d-1) \\
&-\left\lfloor
\frac{n-(d-1)}{t(r+\delta-1)}
\right\rfloor \\
&-\left\lfloor
\frac{
n-2(d-1)
-\left\lfloor \frac{n-(d-1)}{t(r+\delta-1)} \right\rfloor
}{t(r+\delta-1)}
\right\rfloor\\
&+c
}
\\
\tabcite{[\Cref{thm:SingletonboundEAavail}]}
}
\tabularnewline
\hline

Random dim.
&
\makecell[l]{
$n-2\left(\frac{n}{r+1}+\ell\right)$
\\
\tabcite{\cite[Lemma 38]{golowich2025quantum}}
}
&
\makecell[l]{
$n-2\left(\frac{n}{r+\delta-1}(\delta-1)+\ell\right)$
\\
\tabcite{\cite[Lemma 3.3]{guruswami2026quantumhierarchicallocallyrecoverable}}
}
&
\makecell[l]{
$n-2\left(\frac{n}{r+\delta-1}(\delta-1)+\ell\right)+c$
\\
\tabcite{[\Cref{prop:random_eaqlrc}]}
}
\tabularnewline
\hline

QTB dim.
&
\makecell[l]{
$2(\ell-q)\left(1-\frac{2}{r+1}\right)+\epsilon$
\\
$\ell\in[q]_0,\ \epsilon\in[-2,2]$
\\
\tabcite{\cite[Lemma 56]{golowich2025quantum}}
}
&
\makecell[l]{
$1+2(\ell-q)\left(1-\frac{2(\delta-1)}{r+\delta-1}\right)+\epsilon,$
\\
$\ell\in[q]_0$,\\
$\epsilon\in[-2(\delta-1),\,2(\delta-1)]$
\\
\tabcite{\cite[Lemma 5.4]{guruswami2026quantumhierarchicallocallyrecoverable}}
}
&
\makecell[l]{
$(q-1)-2\left\lfloor \frac{\ell}{r+\delta-1} \right\rfloor
+$\\
$2\min\!\left( \ell \bmod (r+\delta-1),\, r \right)$\\$+c$
\\ $\ell\in[q]_0$  \\[0.0em]
\tabcite{[\Cref{thm:EA_QTB}]}
}
\tabularnewline
\hline
\end{tabularx}
\caption{Comparison of bounds on the dimension of  QLRCs from the literature and the entanglement-assisted codes studied in this work. Here $c=\dim(C_1)-\dim(\operatorname{Hull}_{C_2}(C_1))$ denotes the entanglement parameter (see \Cref{rem:EACSS_params}). We also normalize the comparison with prior bounds in \cite{golowich2025quantum} by replacing $r$ with $r+1$ to match the prevailing locality convention.}
\label{table:QLRC_bounds}
\end{table*}
To facilitate comparison, \Cref{table:QLRC_bounds} uses a common parameter convention across the literature. In particular, we write \(c\) for the number of shared maximally entangled pairs (ebits) used in an entanglement-assisted construction. We also normalize the recovery-set notation when quoting results from sources that measure locality by the number of other coordinates needed for repair: in those cases, we replace \(r\) by \(r+1\). In the entanglement-assisted setting, the parameter \(c\) records the amount of shared entanglement available to the encoder and decoder. For constructions of the form \(EACSS(C_1,C_2)\), the value of \(c\) is determined by the failure of the dual-containment condition and may be expressed in terms of the corresponding relative hull. In the check-space convention used for the constructions below, the special case \(EACSS(C,C)\) has \[ c=\dim(C)-\dim(\operatorname{Hull}(C)). \]
The table compares known bounds and construction parameters for QLRCs with the entanglement-assisted variants studied here. The row labeled ``Singleton bound'' gives the corresponding upper bound on the quantum dimension \(k\), while ``Random dim.'' and ``QTB dim.'' refer to dimensions achieved by random constructions and quantum Tamo--Barg-type constructions, respectively.
In many cases, sharper parameters can be obtained by determining the hull $\mathrm{Hull}(C):=C \cap C^{\perp}$ or relative hull $\mathrm{Hull}_{C}(C'):=C' \cap C^{\perp}$ of the underlying classical codes $C$ and $C'$.  The hull or relative hull of the underlying classical codes controls the entanglement consumption and therefore affects the resulting entanglement-assisted dimension. Compared with CSS-type constructions built from the same underlying classical data, the entanglement-assisted framework can allow more flexible parameter choices because it does not require dual containment. When $t=1$ and $c=0$ in \Cref{cor:Singleton2}, we recover the original QLRC Singleton bound of \cite[Theorem 35]{golowich2025quantum} for the $\delta=2$ case and the general $(r,\delta)$ bound from \cite{guruswami2026quantumhierarchicallocallyrecoverable}.  For $t=1$ and $c \geq 0$, our bound is a generalization of the other two Singleton bounds, applicable in the entanglement-assisted setting, where the $c$ ebits relax the dual-containment constraint $C_2^{\perp}\subseteq C_1$, which is not required in the entanglement-assisted setting and therefore gives greater freedom in the choice of LRCs. 
In the entanglement-assisted quantum local recovery setting, we are able to provide a Singleton bound for $(r,\delta)$ EAQLRCs with availability $t$, whereas in the CSS setting previously considered, availability in this disjoint-recovery-set sense is not possible.

The remainder of this paper is organized as follows. \Cref{sec:prelims} establishes notation and contains background necessary for the later sections. \Cref{sec:EALRC} introduces entanglement-assisted quantum locally recoverable codes, and \Cref{sec:random} provides random constructions. Constructions derived from algebraic-geometry codes and curves over finite fields are found in \Cref{sec:constructions}. The paper concludes with a summary and open problems in \Cref{sec:concl}.

\section{Preliminaries}
\label{sec:prelims}

In this section, we establish the notation used throughout the paper and introduce the main concepts needed for the code constructions developed in the subsequent sections. We begin with the necessary background on linear and quantum codes, and then review the key notions and results on local recovery in both the classical and quantum settings.

\subsection{Classical and quantum codes}

\subsubsection{Linear codes}
Let $\Z_{> 0}$ be the set of positive integers, $\mathbb{Z}_{\geq 0} := \Z_{>0} \cup \{0\}$, and $\mathbb{C}$ be the field of complex numbers. For \(t \in \Z_{\geq 0}\), define \([t]_0 := \{0,1,\dots,t-1\}\), and for \( t \in \Z_{>0}\), set \([t] := \{1,\dots,t\}\). Let \(\mathbb{F}_q\) be the finite field with \(q=p^s\) elements, where \(p\) is prime and \(s \in \Z_{> 0}\); let \({\mathbb{F}_q}^*:=\mathbb{F}_q \setminus \{  0 \}\) denote its multiplicative group of nonzero elements. Denote by \(\mathbb{F}_q^n\) the \(n\)-dimensional vector space over \(\mathbb{F}_q\). The support of a vector $u\in \mathbb{F}_q^n$ is \(\mathrm{supp}(u) = \{ j \in [n] \mid u_j \neq 0 \}\).
The set of all \(m \times n\) matrices over \(\mathbb{F}_q\) is written as \(\mathbb{F}_q^{m \times n}\); for \(A \in \mathbb{F}_q^{m \times n}\), the \((i,j)\)-th entry is denoted by \(A_{ij}\). 
For a subset \(I \subseteq [n]\) with \(|I|=\ell\), let \(\pi_I : \mathbb{F}_q^n \to \mathbb{F}_q^\ell\) denote the projection onto the coordinates indexed by \(I\).

Let \(C\) be an \([n,k,d]_q\) linear code over \(\mathbb{F}_q\), meaning a $k$-dimensional subspace of $\F_q^n$ in which any two distinct vectors (called codewords) differ in at least $d$ coordinates and there is a pair of codewords that differ in exactly $d$ coordinates; we omit the subscript $q$ when the field size is clear from context. Here \(n\), \(k\), and \(d\) denote the length, dimension, and minimum distance, respectively. A linear code is called maximum-distance separable (MDS) if and only if $d = n-k+1$, meaning that the Singleton Bound $k+d \leq n+1$ is met with equality. The weight of a codeword \(c \in C\), denoted \(\mathrm{wt}(c)\), is its distance from the zero codeword. Let \(G \in \mathbb{F}_q^{k \times n}\) and \(H \in \mathbb{F}_q^{(n-k) \times n}\) be generator and parity-check matrices of \(C\), respectively, meaning that $C$ is the row space of $G$ and is the nullspace of $H$. The dual code of \(C\) is its orthogonal complement with respect to the standard Euclidean inner product, defined as
\[
C^\perp := \{ u \in \mathbb{F}_q^n \mid u \cdot v = 0 \text{ for all } v \in C \}.
\]
The hull of code $C$ is its intersection with its dual, that is $\mathrm{Hull}(C) = C \cap C^\perp$. The relative hull of a code $C'$ with respect to code $C$ is $\mathrm{Hull}_{C}(C') = C' \cap C^\perp$. The punctured code of $C$ onto coordinates indexed by a subset $I$ of cardinality $\ell$ is the $[\ell,k',d']$ code
\[
C|_I := \pi_I(C) = \{ (c_j)_{j \in I} \mid c = (c_1,\dots,c_n) \in C \} 
\] where $k-(n-\ell) \leq k' \leq k$ and $d-(n-\ell) \leq d' \leq d$.

Another useful way to derive new codes from existing ones is {folding}, in which consecutive coordinates are grouped into symbols over an extended alphabet \cite{GuruswamiRudra2008,GuruswamiXing2012}. Let $C \subseteq \mathbb{F}_q^n$ be a linear code, and suppose $m \in \Z_{>0}$ is such that $m \mid n$. 
Let \(N:=n/m\). For each \(i\in[N]\), define \[ B_i:=\{(i-1)m+1,\ldots,im\}. \]
The \emph{$m$-fold} of $C$ is the code
\[
C^{(m)}
:=
\left\{
\bigl((c_j)_{j\in B_1},\dots,(c_j)_{j\in B_N}\bigr)
:\;
(c_1,\dots,c_n)\in C
\right\}
\subseteq (\mathbb{F}_q^m)^N.
\]
Note that $C^{(m)}$ is linear over $\F_q$, not necessarily over $\F_{q}^m$ (or $\F_{q^m}$) for $m>1$. Under the folding operation, code $C^{(m)}$ is $\F_q$-linear. We say that the code $C^{(m)}$ has dimension $k/m$ over $\F_q^m$, where $k/m$ need not be an integer, acknowledging this is an abuse of notation. Indeed, $\F_q^m$ is not a field itself (though it is isomorphic to $\F_{q^m}$) and hence $C^{(m)}$ is not an $\F_q^m$-vector space; moreover, it may also be the case that the image of $C^{(m)}$ is not linear over $\F_{q^m}$. We say that $C^{(m)}$ has minimum distance $d':=\min \left\{ 
d(u,v): u, v \in C^{(m)}, u \neq v
\right\} \geq \left\lceil d/m \right \rceil$, where $d(c,c'):=\mid \left\{ i \in [N]:  (c_j)_{j\in B_i} \neq (c_j')_{j\in B_i}\right\} \mid $ for $c, c' \in C^{(m)}$. We define the dual of the folded code \(C^{(m)}\) to be the fold of the dual code, \( \left(C^\perp\right)^{(m)}\). It is worth noting that if $C$ happens to be $\F_{q^m}$-linear so that \(C^{(m)}\) is \(F_{q^m}\)-linear after making the identification \(F_q^m\cong F_{q^m}\) via a trace-dual basis, then \( \left(C^\perp\right)^{(m)}=\left(C^{(m)}\right)^\perp\).

The new code constructions defined in this manuscript rely on families of algebraic-geometry (AG) codes. We start with the definition of general evaluation codes. For a subset $S \subseteq \mathbb{Z}_{\ge 0}$, let
\[
\mathbb{F}_q[X]_S
= \left\{ \sum_{i \in S} a_i X^i \;:\; a_i \in \mathbb{F}_q \right\}
\]
denote the space of polynomials for which only monomials $X^i$ with $i \in S$ may have non-zero coefficients. Let $T \subseteq \F_q$ and
\[
\mathrm{ev}_T : \mathbb{F}_q[X]_S \longrightarrow \mathbb{F}_q^{\,\mid T \mid}
\]
denote the evaluation map on $T$ defined by
\(
\mathrm{ev}(f) = \big( f(x) \big)_{x \in T}\). The subscript $T$ is omitted if the set of evaluation points is clear from the context. The most popular family of evaluation codes consists of Reed-Solomon codes, meaning those in which $T= \left\{ \alpha_1,\ldots,\alpha_n \right\}$ is a set of distinct elements of $\F_q$ and $S=[k]_0$. More formally, let \(
\mathcal{P}_{k-1}
=
\{f \in \mathbb{F}_q[x] : \deg(f)<k\}.
\) Then the Reed--Solomon code of length \(n\) and dimension \(k\) is the evaluation code
\[
\operatorname{RS}_k(\alpha_1,\ldots,\alpha_n)
=
\operatorname{ev}(\mathcal{P}_{k-1})
=
\left\{
\bigl(f(\alpha_1),\ldots,f(\alpha_n)\bigr)
:\;
f \in \mathbb{F}_q[x],\ \deg(f)<k
\right\},
\]
which is an \([n,k,n-k+1]_q\) MDS code.

Algebraic-geometry codes (sometimes called AG codes), detailed below, extend the Reed-Solomon framework. Reed--Solomon codes take as evaluation points a collection of affine points on the projective line and use polynomials of bounded degree to define codewords. More general algebraic-geometry codes instead use rational points on a projective curve together with rational functions from a Riemann--Roch space. \rmv{ The Hermitian and Suzuki codes which will be used in \Cref{sec:constructions} are special cases of algebraic-geometry codes coming from curves with many rational points. The fiber product constructions considered there build on this same algebraic-geometry code perspective but, as we will describe, use several compatible covering maps to produce locality and availability simultaneously. With this motivation, we next review the algebraic-geometry code construction. }

Let $X/\mathbb{F}_q$ be a smooth projective curve of genus $g$. Let $D=P_1+\cdots+P_n$
and $G$ be divisors on $X$ such that
$P_1,\dots,P_n$ are pairwise distinct $\mathbb F_q$-rational points of $X$, none of which appear in the support of $G$. The Riemann--Roch space of $G$ is defined as
\[
\mathcal{L}(G) = \{ f \in \F_q(X) : (f) + G \ge 0 \} \cup \{0\},
\]
i.e., for $G \geq 0$, the set of functions whose poles are bounded by $G$. Let $\ell(G) := \dim_{\F_q}\mathcal{L}(G)$. The associated algebraic-geometry code is defined as
\[
C(D,G)
=
\{(f(P_1),\dots,f(P_n)) : f \in \mathcal{L}(G)\}.\] If $2g-2 < \deg(G) < n$, then $C(D,G)$ is an $[n,k,d]$ code with $k = \ell(G)$, and $d \ge n - \deg(G)$. If $G=\alpha P$ for some positive integer $\alpha$ and $P \in X(\F_q)$, then $C(D,G)$ is called a one-point code. Furthermore, if $X$ is the projective line over $\F_q$, then $C$ is a Reed--Solomon code; if $X$ is the Hermitian curve defined by $y^q + y = x^{q+1}$ over $\F_{q^2}$, then $C$ is a Hermitian code; and if $X$ is the Suzuki curve defined by $y^q-y=x^{q_0}(x^q-x)$ over $\F_q$ where $q_0=2^t$ and $q=2q_0^2=2^{2t+1}$ for some positive integer $t$, then $C$ is a Suzuki code. These will all be used in \Cref{sec:constructions} to design EAQLRCs with multiple, disjoint recovery sets for each symbol.

\subsubsection{Quantum codes}
We now turn our attention to the quantum setting, which is the focus of this paper. We review the definitions of CSS codes and entanglement-assisted codes relevant to the results in this paper. For $x=(x_1,\ldots,x_n)\in \F_q^n$, we write $\ket{x}$ for the corresponding computational basis state
\[
\ket{x}=\ket{x_1}\otimes \cdots \otimes \ket{x_n}\in (\C^q)^{\otimes n},
\]
where the computational basis of $\C^q$ is indexed by the elements of $\F_q$. Thus, for $x,y\in \F_q^n$, the ket $\ket{x+y}$ denotes the basis state labeled by the vector sum $x+y$ in $\F_q^n$.

We now recall the definition of a CSS code (see \cite{Calderbank_1996,Steane_1996}). 

\begin{definition}\label{def:CSS}
Let $C',C\subseteq \F_q^n$ be linear codes with parameters $[n,k',d']$ and $[n,k,d]$, respectively, such that $C^\perp\subseteq C'$. The associated quantum \emph{CSS code} is the subspace
\[
\textup{CSS}(C',C)
:=
\operatorname{Span}_{\C}\left\{
\frac{1}{\sqrt{|C^\perp|}}
\sum_{y\in C^\perp}\ket{x+y}
:\ x\in C'
\right\}.
\]
\end{definition}

\begin{remark}
     It is well known that \(\textup{CSS}(C', C)\) is a $[[n,k'+k - n,\min\{\textup{wt}(C'\setminus C^{\perp}), \textup{wt}(C \setminus C'^{\perp})\}]]$ quantum code; see, for example, \cite{Calderbank_1996,gottesman1997stabilizercodesquantumerror, Steane_1996}.
\end{remark}

The CSS construction provides a large and important class of quantum codes when the underlying classical codes satisfy the dual-containment condition $C^\perp \subseteq C'$. Because the dual-containment condition can be too restrictive for many code families of interest, it is natural to use the entanglement-assisted framework, where pre-shared entanglement compensates for the failure of dual containment.
 Unlike standard stabilizer codes and the CSS codes in \Cref{def:CSS}, EAQECCs do not require the dual-containment condition on the underlying classical codes, providing greater flexibility in code constructions. In general, an EAQECC with parameters
\(
[[n,k,d;c]]_q
\)
encodes $k$ logical $q$-ary qudits into $n$ physical qudits, has minimum distance $d$, and uses $c$ shared maximally entangled pairs (ebits) between sender and receiver.
We next review the EAQECC construction from \cite{Galindo2021CorrectionTE} as it applies to CSS codes. The construction uses a pair of $\F_q$ linear codes of equal length, yielding a CSS-like family of EAQECCs. Recall that for two linear codes \(A,B\subseteq \mathbb F_q^n\), we write \[ \operatorname{Hull}_{A}(B):=B\cap A^\perp \] to mean the relative hull of $B$ with respect to $A$.

An entanglement-assisted code can be constructed using two classical codes as follows. Let $C',C \subseteq \mathbb{F}_q^n$ be linear codes with parameters
$[n,k',d']$ and $[n,k,d]$, respectively. Set
\[
c = \dim(C'^\perp)-
\dim \left( \mathrm{Hull}_{C'}(C) \right)
=\dim(C^\perp)-
\dim \left( \mathrm{Hull}_{C}(C') \right).
\]
Choose augmented codes
\(
\widetilde C',\widetilde C \subseteq \mathbb{F}_q^{\,n+c}
\)
such that $\dim(\widetilde C)=k +c$ and $\dim(\widetilde C')=k' +c$ and
\[
\widetilde C^\perp \subseteq \widetilde C'.
\]
The \emph{entanglement-assisted quantum code} associated to $C'$ and $C$ is 
\[
\textup{EAQ}(C',C)
:= \operatorname{Span}_{\C}
\left\{
\frac{1}{\sqrt{|\widetilde C^\perp|}}
\sum_{y\in \widetilde C^\perp}\ket{x+y}
:\;
x\in \widetilde C'
\right\}.
\]

It can be checked that $\textup{EAQ}(C',C)$ is an
\(
[[n,\;k'+k-n+c,\;d;\;c]]
\)
entanglement-assisted quantum code, that is, a code encoding
$k'+k-n+c$ logical qudits into $n$ transmitted qudits with the assistance
of $c$ pre-shared maximally entangled pairs (ebits) between sender and receiver.
Equivalently, the underlying stabilizer acts on $n+c$ qudits, where the first
$n$ qudits are transmitted through the channel and the remaining $c$ qudits are
the ebits pre-shared by the sender and the receiver. Here, 
\[
d \geq
\begin{cases}
\min\{\mathrm{wt}(C'^\perp),\, \mathrm{wt}(C^\perp)\} 
& \text{if } C'^\perp \subseteq C \\[6pt]
\min\big\{
\mathrm{wt}\big(C'^\perp \setminus 
\mathrm{Hull}_{C'}(C)\big),
\mathrm{wt}\big(C^\perp \setminus
\mathrm{Hull}_{C}(C')\big)
\big\}
& \text{otherwise.}
\end{cases}
\]
See, for example, \cite{GalindoHernandoMatsumotoRuano2019,Galindo2021CorrectionTE}
for the EAQ construction over arbitrary finite fields, and
\cite{relative_hulls_24}
for the relative-hull formulation.

Considering 
the $EAQ(C',C)$ constructed above, when \(C^\perp\subseteq C'\), we have that \(C'^\perp\subseteq C\). Therefore \[ C\cap C'^\perp=C'^\perp, \] and hence \[ c=\dim(C'^\perp)-\dim(C\cap C'^\perp)=0. \] Thus, the entanglement-assisted construction recovers the usual CSS setting.
Hence, this definition reflects the connection between CSS codes and their entanglement-assisted counterparts. Up to this point, we have used \(C'\) and \(C\) denote the two classical codes used in the CSS construction; later, we may use \(C_X\) and \(C_Z\) when emphasizing check spaces. At times, it is convenient to consider a parity-check perspective instead, as captured in the next definition. 
For the remainder of the paper, we will refer to the following family of codes as EACSS codes.

\begin{definition}
\label{def:EACSS_check_space}
Let \(C_X,C_Z\subseteq \mathbb F_q^n\) be linear codes with parameters
\([n,k_X,d_X]\) and \([n,k_Z,d_Z]\), respectively. Regard \(C_X\) and
\(C_Z\) as check spaces so that  elements of \(C_X\) define \(X\)-type checks, and
elements of \(C_Z\) define \(Z\)-type checks.
Set
\[
 c
 =k_Z-
 \dim \left(\mathrm{Hull}_{C_X}(C_Z) \right)
 =k_X-
  \dim \left(\mathrm{Hull}_{C_Z}(C_X)\right)
.
\]
Choose augmented check spaces
\(
\widetilde C_Z,\widetilde C_X\subseteq \F_q^{n+c}
\)
such that
\(\dim(\widetilde C_Z)=k_Z\),
 \(\dim(\widetilde C_X)=k_X\),
 and \[\widetilde C_X\subseteq\widetilde C_Z^\perp.
\]
The associated augmented CSS code is
\[
\widetilde Q
:=
\operatorname{Span}_{\C}
\left\{
\frac{1}{\sqrt{|\widetilde C_X|}}
\sum_{y\in \widetilde C_X}|x+y\rangle
:
 x\in \widetilde C_Z^\perp
\right\}
\subseteq (\mathbb C^q)^{\otimes(n+c)}.
\]
The corresponding entanglement-assisted CSS code, denoted
\(EACSS(C_X,C_Z),
\)
is obtained by interpreting the first \(n\) coordinates as the transmitted
qudits and the final \(c\) coordinates as the receiver's halves of the shared
ebits.  
\end{definition}

\begin{remark}
\label{rem:EACSS_params} 
It can be verified that \(EACSS(C_X,C_Z)
\) has parameters
\(
 [[\,n,\; n-k_X-k_Z+c,\; d;\; c\,]]
\) where 
\[ d\ge \min\left\{ \operatorname{wt}\bigl(C_X^\perp\setminus \operatorname{Hull}_{C_X}(C_Z)\bigr), \operatorname{wt}\bigl(C_Z^\perp\setminus \operatorname{Hull}_{C_Z}(C_X)\bigr) \right\}. \]
    \end{remark}
 At times, we are interested in the case in which $C=C_X=C_Z$. Then we see that \(EACSS(C,C)\) is a  
 \( [[n, n-2k+c, d; c]]\) code where 
 \[ d\ge \min\left\{ \operatorname{wt}\bigl(C^\perp\setminus \operatorname{Hull}_{C}(C)\bigr), \operatorname{wt}\bigl(C^\perp\setminus \operatorname{Hull}_{C}(C)\bigr) \right\} \] and $k=k_X=k_Z$.
 In the standard model for EAQECCs, the \(c\) shared ebits are treated as noiseless pre-shared entanglement and are used only as a resource for decoding. Variants in which the receiver's halves of the ebits are also subject to noise have been studied in the literature; see, for example, \cite{AHN201647,PhysRevA.86.032319,PhysRevA.78.012337}. In this work, however, we restrict attention to the standard noiseless-ebit setting, which is the framework relevant for the constructions and bounds developed below.
 
In \cite{Grassl_2022}, the authors provide a Singleton bound for entanglement-assisted codes  which in the case $d -1 \ < n/2$ is  
\[k \leq n - 2(d-1) + c,\] meaning the entanglement-assisted quantum Singleton bound differs from the quantum Singleton bound by the amount of shared entanglement $c$. In this work, we consider the setting in which $d-1 < n/2$, since otherwise $3d-3-n$ grows faster than $n-d+1$, yielding a much smaller upper bound on $k$. In \Cref{sec:EALRC}, we will extend this bound to the case of entanglement-assisted LRCs. 

\subsection{Local recovery}

\subsubsection{Classical locally recoverable codes}
We begin in the classical setting because the quantum codes in this paper are built from pairs of classical LRCs. The key idea behind locality is that each coordinate of a codeword should be recoverable from a small subset of other coordinates, while availability requires the existence of several such recovery sets. In algebraic-geometry constructions, these recovery sets often arise from the geometry of the evaluation set, especially from fibers of suitable morphisms to auxiliary curves. We therefore review both the coding-theoretic definition of locality and the geometric setup that will later allow us to construct explicit families of EAQLRCs with availability. 

We begin with the standard notions of locality, $(r,\delta)$ locality, and
availability for linear codes; see, for example,
\cite{LRC_intro,KPLK14,availability_intro}. A linear code $C \subseteq \mathbb{F}_q^n$ is a \emph{locally recoverable code (LRC)} with locality $r$ if for all $i \in [n]$ there exists a set of coordinates \( I(i) \subseteq [n]\) containing $i$ 
with size bounded as \( |I(i)| \le r+1,\) such that $c_i$ is determined by $\{c_j : j \in I(i) \setminus \{ i \} \}$ for all $c \in C$. Such a code recovers exactly one erasure per recovery set (that is, $\delta = 2$).

We say that $C$ is an $(r,\delta)$ LRC with availability $t$ if, for every coordinate $i \in[n]$, there exists a set
\(
J_i=\{i_1,\ldots,i_{\delta-1}\}\subseteq[n]
\)
with $i\in J_i$, and there exist $t$ recovery sets
\(
I_1(i),\ldots,I_t(i)\subseteq[n]
\)
such that, for all $j\ne j'$,
\[
I_j(i)\cap I_{j'}(i)=J_i,
\]
and, for each $j\in[t]$,
\[
|I_j(i)|\le r+\delta-1,\qquad
d(C|_{I_j(i)})\ge \delta,\qquad
\dim(C|_{I_j(i)})\le r.
\]
Unlike the classical definition, we allow recovery sets to intersect on the $\delta-1$ coordinates that the local code is capable of correcting. We emphasize that we use the following set-erasure notion of availability, generalizing the more common single erasure setting. Let \(J\subseteq[n]\) with \(|J|\le \delta-1\). We say that \(C\) has \((r,\delta)\) locality with availability \(t\) for such erasure sets if there exist sets \[ I_1(J),\ldots,I_t(J)\subseteq[n] \] such that \(J\subseteq I_a(J)\), \(|I_a(J)|\le r+\delta-1\), and \[ I_a(J)\cap I_b(J)=J \qquad\text{for all }a\ne b. \] Moreover, for each \(a\), the punctured code \(C|_{I_a(J)}\) has minimum distance at least \(\delta\) and dimension at most \(r\).

For comparison, we recall the Singleton-type bound for the standard coordinate-availability setting
\cite{kadiyam_das_2020_lrc_with_availibility}. Let \(C\) be an \([n,k,d]\) linear code over \(\mathbb{F}_q\)
with locality parameters \((r,\delta)\) and availability \(t\). Then 
\[
d \leq n-k+1-
\left(
\left\lceil \frac{t(k-1)+1}{\,t(r-1)+1\,} \right\rceil
-1
\right)(\delta-1).
\]
Setting $t=1$ recovers the 
Singleton bound presented in \cite{prakash2012optimallinearcodeslocalerrorcorrection}. We will apply this bound only in situations where the hypotheses agree with the standard coordinate-availability setting. An LRC is called optimal when the minimum distance of the code achieves the Singleton bound. Optimal classical $(r,\delta)$ LRCs were defined by Tamo and Barg in \cite{Barg_2015} for $\delta=2$ and generalized to any integer $\delta \ge 2$ in \cite{ballentine.barg.vladut.hlrc}.

Next, we discuss quantum local recovery. We use $\mathrm{Tr}$ to denote the usual trace, and for $B \subseteq A$, 
$\mathrm{Tr}_{A\setminus B}$ denotes the partial trace over the tensor factors
indexed by $A\setminus B$. Let $M(A)$ denote the set of density matrices on a set $A$ of qudits, that is,
\[
M(A) := \left\{ \rho : (\mathbb{C}^q)^{\otimes A} \to (\mathbb{C}^q)^{\otimes A} \;\middle|\; \rho \ge 0,\ \mathrm{Tr}(\rho)=1 \right\}.
\]
For a pure state $\ket{\psi} \in (\mathbb{C}^q)^{\otimes A}$, we write $\psi = \ket{\psi}\bra{\psi} \in M(A)$. For $\rho \in M(A)$ and $B \subseteq A$, the reduced density matrix is denoted by
\[
\rho_B := \mathrm{Tr}_{A\setminus B}(\rho) \in M(B).
\]
A quantum channel $\mathcal{K}$ from qudits $A$ to qudits $B$ is a completely positive, trace-preserving map
\[
\mathcal{K} : M(A) \to M(B).
\]
Since all Hilbert spaces considered here are finite-dimensional, every quantum
channel admits a Kraus representation
\[
K(\rho)=\sum_{\nu} K_\nu \rho K_\nu^\dagger,
\quad
\text{with }
\sum_\nu K_\nu^\dagger K_\nu = I.
\]

\subsubsection{Quantum locally recoverable codes}

Quantum local recovery was introduced in \cite{golowich2025quantum} in the
setting of a single erased qudit. In that formulation, each coordinate $i$
admits a local recovery channel acting on a small neighborhood of $i$ and
reconstructing the erased qudit from the remaining qudits in that neighborhood.
Following the convention adopted throughout this paper, we write the locality
parameter as $r$ when the corresponding recovery set has size at most $r+1$,
in agreement with the classical literature. 
The same work also generalizes the definition to recovery sets capable of
correcting multiple erasures. Given $\emptyset \neq I \subsetneq
J \subseteq [n]$, a quantum code $Q \subseteq (\mathbb{C}^q)^{\otimes n}$ is
said to be an \emph{$(I,J)$ QLRC} if there exists a
trace-preserving quantum operation
\[
\mathcal{R}_{Q,I}^J,
\]
which acts only on the qudits corresponding to $J$ and leaves all other qudits
unchanged, such that
\[
\mathcal{R}_{Q,I}^J \circ \tau_I\bigl(\ket{\varphi}\!\bra{\varphi}\bigr)
=
\ket{\varphi}\!\bra{\varphi},
\]
for all $\ket{\varphi}\in Q$.
Furthermore, let $r,\delta \in \mathbb{Z}_{>0}$ with $\delta \ge 2$. A quantum error-correcting code 
${Q} \subseteq (\mathbb{C}^q)^{\otimes n}$ is called an \emph{$(r,\delta)$ QLRC} if for each coordinate $i \in [n]$, there exists a subset $J \subseteq [n]$ containing $i$ such that 
\[
|J| \le r + \delta - 1,
\]
and for all subsets $I \subseteq J$ with $|I| \le \delta - 1$, the code ${Q}$ is an $(I,J)$ QLRC.

Given two copies of classical LRCs, we can construct a QLRC. To do so, we make use of the parity-check spaces, a choice made out of convenience that does not inhibit the generality of the approach. This was first shown for the single erasure recovery case in \cite[Corollary 34]{golowich2025quantum} and generalized to the multiple erasure recovery setting in \cite[Proposition 28]{galindo2024quantumrdeltalocallyrecoverablecodes}. In particular, 
the quantum code \(\operatorname{CSS}(C_X,C_Z)\) is an
\((r,\delta)\) \emph{QLRC} if, for every coordinate
\(i\in[n]\), there exist \(\delta-1\) pairs of parity checks $(c_{x,j},c_{z,j})\in C_X^\perp \times C_Z^\perp, j \in [\delta-1]$
such that $i\in \operatorname{supp}(c_{x,j})\cap \operatorname{supp}(c_{z,j})$
and $\left|\operatorname{supp}(c_{x,j})\cup
\operatorname{supp}(c_{z,j})\right|
\le r+\delta-1$ for each \(j \in [\delta-1]\).

In the next section, we extend quantum local recovery to the entanglement-assisted setting.

\section{Entanglement-assisted codes with local recovery} \label{sec:EALRC}

In this section, we introduce EAQLRCs with availability by extending the notion of quantum local recovery to the entanglement-assisted setting. In contrast to the CSS framework, this approach does not require a dual-containing pair of classical codes, and therefore allows greater flexibility in constructing QLRCs from pairs of classical LRCs.

\begin{definition} \label{def:EAQLRC}
Let $E_B$ denote the receiver's set of ebits so that $|E_B| = c$. The 
entanglement-assisted quantum code $Q=\textup{EACSS}(C_1,C_2)
\subseteq (\mathbb{C}^q)^{\otimes n}$
is an {entanglement-assisted quantum locally recoverable code (EAQLRC)
with locality $r$} provided that it is 
constructed from two classical linear codes $C_1,C_2\subseteq \mathbb{F}_q^n$ together with a set of recovery channels $\textup{Rec}_i$ satisfying the following properties:
\begin{enumerate}
    \item For each coordinate $i\in[n]$, there exists a set
$R_i\subseteq[n]$ with $i\in R_i$ and $|R_i|\le r+1$ such that
\[
\mathrm{Rec}_i:
M\!\left((R_i\setminus\{i\})
\cup E_B \right)
\longrightarrow
M\!\left(R_i \cup E_B\right)
\] with input subsystems $(R_i\setminus\{i\})\cup E_B$
and output subsystems $R_i\cup E_B$.

\item For every code state $\rho\in Q$, let $\widetilde{\rho}$ denote the corresponding joint state of the transmitted qudits together with the receiver's ebits. Then 
\[
(\mathrm{Rec}_i\otimes \mathrm{id}_{[n]\setminus R_i})
\left(\widetilde{\rho}_{[n]\setminus\{i\},\,E_B}\right)
=
\widetilde{\rho}_{[n], E_B}.
\]
\end{enumerate}
\end{definition}

\begin{remark}
The EAQLRC has parameters $[[n,k,d;c]]$,
where $k=n-(k_1+k_2)+c$, $c=\dim(C_1)-\dim(\mathrm{Hull}_{C_2}(C_1)),$ and $k_i=\dim(C_i)$ for $i=1,2$ and has locality $r$. When $c=0$ and $E_B=\emptyset$, \Cref{def:EAQLRC} recovers the definition of QLRC given in \cite{golowich2025quantum}.
\end{remark}

Next, we generalize the definition of $(r,\delta)$ QLRC introduced in \cite{galindo2024quantumrdeltalocallyrecoverablecodes} to the availability setting using EAQECCs. 

\begin{definition}\label{def:r-delta-QLRC}
Let $r,\delta \in \mathbb{Z}_{>0}$ with $\delta \ge 2$. 
An entanglement-assisted quantum error-correcting code 
$Q \subseteq (\mathbb{C}^q)^{\otimes n}$ 
is called an \emph{$(r,\delta)$ EAQLRC} if for each coordinate 
$i \in [n]$, there exists a subset 
$\Delta \subseteq [n]$ containing $i$ such that 
\[
|\Delta| \le r + \delta - 1,
\]
and for all subsets $R \subseteq \Delta$ with 
$|R| \le \delta - 1$, the code $Q$ is an $(R,\Delta)$ QLRC. \\

We say that $Q$ has availability $t$ if, for every $i\in[n]$, there exists a set \( J_i=\{i_1,\ldots,i_{\delta-1}\}\subseteq[n] \) with $i\in J_i$, and there exist $t$ recovery sets \( \Delta_1,\ldots,\Delta_t \subseteq[n] \) such that for all \(a\ne b\), \[ \Delta_a\cap \Delta_b=J_i\] and for all \( \ell\in[t]\), \[ |\Delta_\ell|\le r+\delta-1. \] Moreover, for every subset $R_\ell\subseteq \Delta_\ell$ with $|R_\ell|\le \delta-1$, the code $Q$ is an $(R_\ell,\Delta_\ell)$ QLRC.
\end{definition}

The following result provides a sufficient condition for an EAQECC to be an $(r, \delta)$ EAQLRC with availability $t$. As in the stabilizer-based recovery framework of \cite[Proposition~33]{golowich2025quantum}, local recovery must act correctly on the global
encoded state. Accordingly, we assume that the
erasure-recovery channels of the punctured EACSS
codes also satisfy the global recovery condition
of \Cref{def:EAQLRC}, with access to the
receiver's entanglement resources.

\begin{proposition}[$(r,\delta)$ EAQLRC with availability]\label{cor:EA-Qc}
Let $C_1, C_2 \subseteq \mathbb{F}_q^n$ be two classical $(r,\delta)$ LRCs 
and let $Q = \textup{EACSS}(C_1, C_2)$. 
If for every $i \in [n]$ there exists a set $\Delta \subseteq [n]$ with 
$i \in \Delta$ and $|\Delta| \leq r + \delta - 1$ such that $Q_i := \textup{EACSS}(C_1|_{\Delta}, C_2|_{\Delta})$
has minimum distance at least $\delta$ and 
the local erasure-recovery
channels of $Q_i$ are compatible with
the global code $Q$ and its receiver-side entanglement
resources, in the sense that they
correct the corresponding erasures of every
joint encoded state of $Q$, while acting
trivially on the coordinates outside $\Delta$, then $Q$ is an $(r,\delta)$ EAQLRC. 
If moreover $C_1$ and $C_2$ have availability $t$, that is, there exist
sets $\Delta_1, \ldots, \Delta_t$ for each $i \in [n]$ with 
$\Delta_\ell \cap \Delta_{\ell'} = J_i$ for all $\ell \neq \ell'$, and if $\textup{EACSS}(C_1|_{\Delta_\ell}, C_2|_{\Delta_\ell})$
has minimum distance at least $\delta$ for every $i \in [n]$ and every  $\ell \in [t]$, then $Q$ has availability $t$.
\end{proposition}

\begin{proof} Fix $i \in [n]$ and let $\Delta$ be the recovery set given by hypothesis.
Then $Q_i := \textup{EACSS}(C_1|_{\Delta}, C_2|_{\Delta})$
has minimum distance at least $\delta$ by hypothesis, and therefore
corrects any erasure pattern of size at most $\delta - 1$ within $\Delta$
using $E_B$ and by the compatibility assumption, these local
recovery channels also recover the corresponding
erasures of the global code $Q$ and its entanglement resources, leaving all
coordinates outside $\Delta$ unchanged.
Hence $Q$ is an $(r,\delta)$ EAQLRC. Now doing the same procedure for each recovery set $\Delta_\ell$ separately gives
$(R, \Delta_\ell)$ local recovery for every $R \subseteq \Delta_\ell$
with $|R| \leq \delta - 1$ and every $\ell \in [t]$. Since
$\Delta_\ell \cap \Delta_{\ell'} = J_i$ for $\ell \neq \ell'$, the
availability condition of \Cref{def:r-delta-QLRC} is satisfied.
\end{proof}

We now have the following corollary. We say $C$ is an optimal LRC if its parameters satisfy the Singleton bound for LRCs with equality. 

\begin{corollary}\label{cor_optimal}
Let $C$ be an optimal classical $(r,\delta)$ LRC with $r \geq \delta-1$ 
and recovery sets of size exactly $r+\delta-1$, and let 
$Q = \textup{EACSS}(C,C)$. Assume that the local EACSS recovery channels associated with these recovery sets satisfy the conditions in 
\Cref{cor:EA-Qc}. Then $Q$ is an $(r,\delta)$ EAQLRC. 
If $C$ has availability $t$ then $Q$ has availability $t$.
\end{corollary}

\begin{proof}
For each recovery set $\Delta$ of size $r+\delta-1$, the local 
restricted code $C|_\Delta$ has dimension $r$ and distance at least $\delta$. 
By the Singleton bound
$$d \leq (r+\delta-1)-r+1 = \delta,$$ 
so $C|_\Delta$ is MDS with parameters $[r+\delta-1, r, \delta]$. By MDS 
duality, $C|_\Delta^\perp$ has distance $r+1\geq\delta$. Hence 
$\textup{EACSS}(C|_\Delta,C|_\Delta)$ has minimum distance at least 
$\delta$ by \Cref{rem:EACSS_params} and since the compatibility condition holds by hypothesis, the result follows from
\Cref{cor:EA-Qc}. Since this holds for every 
recovery set $\Delta$ of $C$, including all $t$ recovery sets 
$\Delta_1,\ldots,\Delta_t$ when $C$ has availability $t$, the result 
follows from \Cref{cor:EA-Qc}.
\end{proof}

We now provide a Singleton-like bound on EAQLRCs. In particular, the next theorem extends existing Singleton-like bounds for QLRCs to the entanglement-assisted setting for those defined using a CSS-like construction, where the entanglement parameter $c$ also plays a role in the tradeoff among length, dimension, and local recoverability.

\begin{theorem}[Singleton-like bound for EAQLRCs with availability] \label{thm:SingletonboundEAavail}
Let $Q$ be an EAQLRC of block length $n$, dimension $k>0$, distance $d$, and locality parameters $(r,\delta)$ with availability $t$. Define
\[
s_1 \;:=\; \left\lfloor \frac{n-(d-1)}{t(r+\delta-1)} \right\rfloor
\]
and
\[
s_2 \;:=\; \left\lfloor \frac{\,n-2(d-1)-s_1\,}{t(r+\delta-1)} \right\rfloor.
\]
Assume that $s_2 \geq 0$. Then
\[
k \;\le\; n - 2(d-1) - s_1 - s_2 + c.
\]
\rmv{Equivalently,
\[
k \;\le\; n - 2(d-1)
      - \left\lfloor \frac{n-(d-1)}{t(r+\delta-1)} \right\rfloor
      - \left\lfloor \frac{\,n-2(d-1)-\left\lfloor \frac{n-(d-1)}{t(r+\delta-1)} \right\rfloor\,}{t(r+\delta-1)} \right\rfloor
      + c.
\]}
\end{theorem}

\begin{proof}
We follow the same general strategy as in the QLRC case, though in this setting we must keep track of the entropy contribution coming from the shared entanglement. We construct two disjoint families of locally recoverable coordinates, together with two sets of size $d-1$ that play the role of globally correctable erasures, and then apply an entropy argument to the resulting partition of the physical qudits.

Let $E$ denote the $c$ qudits corresponding to the halves of the shared maximally entangled pairs held by the decoder. Let $A$ be a $k$-qudit reference system maximally entangled with the logical space. The encoding map acts on one half of the entangled logical state together with the encoder's halves of the $c$ ebits, and the resulting global state on $AQE$ is pure.

We partition the code coordinates as
\[
Q = S_1 \sqcup V_1 \sqcup S_2 \sqcup V_2 \sqcup W.
\]
For notational convenience, we denote $A \sqcup B$ by $AB$. We first construct $S_1$ and $V_1$, and then repeat the argument on the remaining coordinates to obtain $S_2$ and $V_2$.
 To construct this partition, we first construct $S_1,V_1$ as follows:
  \begin{enumerate}
  \item Initialize sets $S_1=\emptyset$ and $T_1=\emptyset$.
  \item\label{it:setS1} Repeat the following step \[ s_1=\left\lfloor\frac{n-(d-1)}{t(r+\delta-1)}\right\rfloor \] times. Choose some \[ i\in Q\setminus S_1 T_1, \] add \(i\) to \(S_1\), and let \[ I_{i,1},\ldots,I_{i,t} \] be the \(t\) available recovery sets for \(i\). Add to \(T_1\) all coordinates in the union of these recovery sets, excluding \(i\) and any coordinates already placed in \(S_1\): \[ T_1\leftarrow T_1\cup \left( \bigcup_{\ell=1}^t I_{i,\ell}\setminus(\{i\}\cup S_1) \right). \]
  \item\label{it:setV1} Set $V_1$ to be any $(d-1)$-element subset of $Q\setminus S_1 T_{1}$.
  \end{enumerate}
The above procedure is guaranteed to successfully output \(S_1,V_1\) with \(\mid S_1\mid =s_1\) and \(\mid V_1\mid =d-1\). Indeed, each iteration in Step~\ref{it:setS1} adds one element to \(S_1\) and adds at most \(t(r+\delta-1)\) coordinates to \(S_1 T_1\). Therefore, after all \(s_1\) iterations, \[ \mid Q\setminus(S_1T_1)\mid  \ge n-t(r+\delta-1)s_1 \ge d-1. \] Thus there are enough remaining coordinates to choose \(V_1\subseteq Q\setminus S_1T_1\) with \(\mid V_1\mid =d-1\).

  Let $Q_1=Q\setminus S_1V_1$. For a code state $\rho\in Q$, given access to the joint state
$\widetilde{\rho}_{Q_1E_B}$ on the components
in $Q_1$ together with the receiver's ebits
$E_B$, we may by construction pass through
all $i\in S_1$ one-by-one in the same order
$S_1$ was constructed, and apply
$\textup{Rec}_i$ at each step to recover
$\widetilde{\rho}_{Q_1S_1E_B}$.
Then, as $Q\setminus Q_1S_1=V_1$, which
has size $d-1$, a global entanglement-assisted
erasure recovery channel recovers the
entire encoded state on $QE_B$ from
$\widetilde{\rho}_{Q_1S_1E_B}$.

  The joint system \(Q_1E_B\) can be used to recover the entire encoded state on \(QE_B\). It follows that the reduced density matrix $\rho_{Q\setminus Q_1}=\rho_{S_1V_1}$ is the same for all code states $\rho \in Q$. It is clear that $|Q_1| \geq d-1$ by the hypothesis that $s_2 \geq 0,$ which gives us that $n-2(d-1)-s_1 \geq 0$ and implies that $|Q_1| = n - s_1 - (d-1) \geq d-1$.  
  
  We then construct $S_2,V_2$ using a similar procedure as above, but on the restriction to components in $Q_1$:
  \begin{enumerate}
  \item Initialize sets $S_2=\emptyset$ and $T_2=\emptyset$.
  \item\label{it:setS2} Repeat the following step \[ s_2= \left\lfloor \frac{\mid Q_1\mid -(d-1)}{t(r+\delta-1)} \right\rfloor = \left\lfloor \frac{n-2(d-1)-s_1}{t(r+\delta-1)} \right\rfloor \] times. Choose some \[ i\in Q_1\setminus S_2 T_2, \] add \(i\) to \(S_2\), and let \[ I_{i,1},\ldots,I_{i,t} \] be the \(t\) available recovery sets for \(i\). Add to \(T_2\) all coordinates in the union of these recovery sets, excluding \(i\) and any coordinates already placed in \(S_2\): \[ T_2\leftarrow T_2\cup \left( \bigcup_{\ell=1}^t I_{i,\ell}\setminus(\{i\}\cup S_2) \right). \]
  \item\label{it:setV2} Set $V_2$ to be any $(d-1)$-element subset of $Q_1\setminus S_2 T_{2}$.
  \end{enumerate}
The procedure above is guaranteed to successfully output \(S_2 \), \(V_2\) with \( \mid S_2 \mid=s_2\) and \( \mid V_2 \mid=d-1\). Indeed, each iteration of Step~\ref{it:setS2} adds one element to \(S_2\) and adds at most \(t(r+\delta-1)\) coordinates to \(S_2T_2\). Hence, after all \(s_2\) iterations, \[ \mid Q_1\setminus S_2 T_2 \mid \ge \mid Q_1 \mid-t(r+\delta-1)s_2 \ge d-1. \] Thus we may choose \(V_2\subseteq Q_1\setminus(S_2\ T_2)\) with \(\mid V_2 \mid =d-1\). Letting $Q_2=Q\setminus S_2V_2$, then by similar reasoning as above for $Q_1$, the joint system \(Q_2E_B\) can be used to recover the entire encoded state on \(QE_B\) so all code states $\rho\in Q$ have the same reduced density matrix $\rho_{Q\setminus Q_2}=\rho_{S_2V_2}$.

  Now let $W=Q\setminus S_1V_1S_2V_2$, so that we have our desired partition $Q=S_1V_1S_2V_2W$. Let $q$ be the local dimension of $C$, and let $A$ be a set of $k$ additional qudits of local dimension $q$. Define
  \begin{equation}
   \ket{\psi}_{AQE}
= \frac{1}{\sqrt{q^k}}\sum_{x\in[q]_0^k}\ket{x}_A\otimes(\textup{Enc}\ket{x})_{QE}
  \end{equation}
  to be the state obtained by applying the encoding map $\textup{Enc}$ of $C$ to one register from a maximally entangled pair of message states $\sum_{x\in[q]^k}\ket{x}\otimes\ket{x}$. Let $\psi=\ket{\psi}\bra{\psi}$ be the associated density matrix.

 By $(r,\delta)$ locality, for each recovery set $I_{ij}$, the punctured code $C|_{I_{ij}}$ has minimum distance at least $\delta$. Hence any $\delta-1$ erasures within $I_{ij}$ can be recovered from the remaining coordinates in $I_{ij}$. By construction, for each $i \in S_b$, all coordinates of $I_{ij}\setminus\{i\}$
are contained in $Q\setminus (S_b \cup V_b)$, except for at most $\delta-1$
coordinates. Therefore $i$ can be recovered. Thus, for $b=1,2$, the systems in $Q\setminus (S_b \cup V_b)$ determine the full code state. Since $|V_b|=d-1$, the distance  condition implies the reduced density matrices on $S_b \cup V_b$ are independent of the encoded state. Now incorporate entanglement assistance. For any $B \subseteq AQE$, let $S(B)$ denote the von Neumann entropy (base $q$). Since the erasures of \(S_1\cup V_1\) and \(S_2\cup V_2\) are correctable with access to \(E_B\), the reference system \(A\) is decoupled from these erased systems. Using the same entropy identities as in the unassisted case in the proof of \cite[Theorem 35]{golowich2025quantum} but now applied to the pure state on $A Q E$, we obtain
\[
S(A) + S(S_1 \cup V_1)
= S(A \cup S_1 \cup V_1)
= S(S_2 \cup V_2 \cup W \cup E)
\le S(S_2 \cup V_2) + S(W) + S(E),
\]
and symmetrically,
\[
S(A) + S(S_2 \cup V_2)
= S(A \cup S_2 \cup V_2)
= S(S_1 \cup V_1 \cup W \cup E)
\le S(S_1 \cup V_1) + S(W) + S(E).
\]

Here, the inequalities follow from subadditivity of entropy. 
Adding the two inequalities, we obtain
\[
S(A) \le S(W) + S(E).
\]

Since $A$ is maximally entangled with the logical space, $S(A) = k.$
Moreover, $S(E)=c$, because $E$ consists of $c$ maximally mixed qudits arising from $c$ ebits. Finally, as in the unassisted case, $S(W) \le |W|$. 
Therefore, we obtain 
\[
k \le |W| + c.
\]

  Meanwhile, $S(W)\leq|W|=n-|S_1 \cup V_1 \cup S_2 \cup V_2|=n-2(d-1)-s_1-s_2$. Thus
  \begin{equation*}
    k = S(A) \leq S(W) + S(E) \leq n-2(d-1)-\left\lfloor\frac{n-(d-1)}{t(r+\delta-1)}\right\rfloor-\left\lfloor\frac{n-2(d-1)-\left\lfloor\frac{n-(d-1)}{t(r+\delta-1)}\right\rfloor}{t(r+\delta-1)}\right\rfloor + c,
  \end{equation*}
  as desired.
\end{proof}

\begin{remark} Omitting the second last term in the bound, we obtain \[
k \leq n-2(d-1)
-\frac{n-(d-1)}{t(r+\delta-1)}
+1+c.
\]
Equivalently,
\[
d-1 \leq
\frac{
\left(1-\frac{1}{t(r+\delta-1)}\right)n-k+c+1
}{
2-\frac{1}{t(r+\delta-1)}
}.
\]
Thus, for fixed locality parameters \(r,\delta\) and  fixed availability \(t\), even as the rate \(k/n\to 0\), the relative distance of an EAQLRC satisfies
\[
\frac{d}{n}
\leq
\frac{1}{2}
-\Omega\!\left(\frac{1}{t(r+\delta-1)}\right) + \mathcal{O}\left(\frac{c}{n}\right).
\]

\end{remark}

\begin{remark} We note that $s_2 \geq 0$ is necessary to ensure that the constructions 
of $S_2$ and $V_2$ exist, and holds whenever 
\[
n - 2(d-1) \;\geq\; \left\lfloor \frac{n-(d-1)}{t(r+\delta-1)} \right\rfloor = s_1.
\]
This condition is satisfied for all constructions presented in this 
paper. When $s_2 < 0$, the bound does not apply; such cases arise 
when the distance is large relative to the length, specifically when 
$d - 1$ exceeds the threshold above.
    
\end{remark}

Under a stronger disjointness assumption on the local recovery sets, the Singleton-like bound in  \Cref{thm:SingletonboundEAavail} can be sharpened. The following result follows from the definition of EAQLRCs and \cite[Theorem 31]{galindo2024quantumrdeltalocallyrecoverablecodes}. It also applies to QLRCs with availability one.

\begin{corollary}\label{cor:Singleton2}
Let \(Q\) be an EAQLRC of block length \(n\), dimension \(k\), and distance \(d\), with locality parameters \((r,\delta)\). Assume that for each coordinate
\(i \in [n]\), there exists a recovery set \(J(i)\) containing \(i\) such that $|J(i)| = r+\delta-1.$
Furthermore,  assume that for all \(i,j \in [n]\), either $J(i)=J(j)$ or $J(i)\cap J(j)=\emptyset.$
Then
\[
k \leq
n-2(d-1)-\left\lfloor\frac{n-(d-1)}{(r+\delta-1)}\right\rfloor-\left\lfloor\frac{n-2(d-1)-\left\lfloor\frac{n-(d-1)}{(r+\delta-1)}\right\rfloor}{(r+\delta-1)}\right\rfloor + c.
\]
\end{corollary}

From the previous result, one can bound $(r,\delta)$ EAQLRCs with stronger assumptions.

\section{Random constructions of EAQLRCs} \label{sec:random}
In this section, we provide EAQLRC constructions using a Vandermonde matrix similar to one constructed  in \cite[Construction 1]{guruswami2026quantumhierarchicallocallyrecoverable}. 

\begin{construction}\label{cons:construction2} Fix a block length $n,$ a locality parameter $r$, a distance parameter $\delta \leq r$ and a finite field $\mathbb{F}_q$ with $q \geq n$. Assume that $m := n/(r + \delta -1) \in \mathbb{Z}$. We will define a random EAQLRC by using the entanglement-assisted construction giving rise to $Q = \textup{EACSS}(C_X,C_Z)$, which is sampled as follows. Initialize the matrix $G_i \in \mathbb{F}_q^{(\delta -1)\times (r+\delta - 1)}$ as the Vandermonde matrix \[G_i = \begin{bmatrix} 
1 & 1 & \ldots & 1\\
\alpha_{i,1} & \alpha_{i,2} & \ldots & \alpha_{i,r+\delta - 1}\\
\vdots & \vdots & \ddots & \vdots \\
\alpha_{i, 1}^{\delta -2} & \alpha_{i,2}^{\delta - 2} & \ldots & \alpha_{i, r + \delta - 1}^{\delta - 2}\\
\end{bmatrix}\] where for each $i = 1, \ldots m$, every $\alpha_{i,j} \in \F_q$ is distinct for each $j = 1, \ldots, r + \delta -1$. Define $H_X$ as 
\[H_X = \begin{bmatrix}
    G_1 & & & \\
    & G_2 & & \\
    & & \ddots & \\
    & & & G_m
\end{bmatrix}.\] Let $\beta_{i,l} = \prod_{j \neq l} \frac{1}{\alpha_{i,l} - \alpha_{i,j}}$ for $l = 1, \ldots, r+ \delta - 1$. Define $H_i$ as the matrix 

\[H_i = \begin{bmatrix} 
1 & 1 & \ldots & 1\\
\alpha_{i,1} & \alpha_{i,2} & \ldots & \alpha_{i,r+\delta - 1}\\
\vdots & \vdots & \ddots & \vdots \\
\alpha_{i, 1}^{r - 1} & \alpha_{i,2}^{r - 1} & \ldots & \alpha_{i, r + \delta - 1}^{r - 1}\\
\end{bmatrix} \begin{bmatrix}
    \beta_{i, 1} & & & \\
    & \beta_{i,2} & & \\
    & & \ddots & \\
    & & & \beta_{i, r + \delta -1 }
\end{bmatrix}\] and define $H_Z$ as 

\[H_Z = \begin{bmatrix}
    (H_1)_{\delta -1} & & & \\
    & (H_2)_{\delta - 1} & & \\
    & & \ddots & \\
    & & & (H_m)_{\delta - 1}
\end{bmatrix},\] where $(H_i)_{\delta -1}$ represents the first $\delta -1$ rows of $H_i$. 
We then add $\ell$ random rows to $H_X$ and $H_Z$, as follows:
\begin{enumerate}
    \item for $i = 1, \ldots, \ell$:
        \begin{enumerate}
            \item let $v$ be a uniformly random vector sampled from $\F_q^n$ such that $v \not\in \mathrm{rowspan}(H_X)$.
    
            \item append $v$ to $H_X$ and rename the resulting matrix as $H_X$ i.e. $H_X = \begin{bmatrix}
                H_X\\v
            \end{bmatrix}$.
        \end{enumerate}
    \item for $i = 1, \ldots, \ell$:
        \begin{enumerate}
            \item let $v$ be a uniformly random vector sampled from $\F_q^n$ such that $v \not\in \mathrm{rowspan}(H_Z)$.

            \item append $v$ to $H_Z$ and rename the resulting matrix as $H_Z$  i.e. $H_Z = \begin{bmatrix}
            H_Z\\v
            \end{bmatrix}$.
        \end{enumerate}
\end{enumerate}
We have sampled matrices $H_X, H_Z \in \mathbb{F}_q^{(m(\delta-1) + \ell) \times n}$. Let $C_X = \textup{ker}(H_X)$ and $C_Z = \textup{ker}(H_Z)$ and $Q = \textup{EACSS}(C_X,C_Z)$.
\end{construction} 

The following result describes the properties of the quantum code $Q$. 

\begin{proposition} [Random $(r, \delta)$ EAQLRC]
\label{prop:random_eaqlrc} Fix positive integers $n$ and $r$ along with a finite field $\mathbb{F}_q$ such that $q \geq n$ and an integer $\delta \leq r$ such that $m := n/(r + \delta -1) \in \mathbb{Z}$.  
Then the code $Q$ presented in \Cref{cons:construction2} is an EAQLRC with locality $r$ and dimension $n - 2(m(\delta - 1)+\ell) + c$, where $c$ is as in \Cref{def:EACSS_check_space} and other parameters given by \Cref{rem:EACSS_params}. Each  coordinate $i \in [n]$
has recovery set $J_i=
\{
a(r+\delta-1)+1,\ldots,
a(r+\delta-1)+r+\delta-1
\},$
where $a=\left\lceil \frac{i}{r+\delta-1}\right\rceil-1.$
\end{proposition}

\begin{proof} For each coordinate $i \in [n]$, let $a$ be as defined in the theorem statement. The proof follows directly from \cite[Lemma 3.3]{guruswami2026quantumhierarchicallocallyrecoverable} as the construction of $G_{a+1}$ and $(H_{a+1})_{\delta -1}$ have $\delta-1$ rows and the columns that correspond to coordinates in $J_i$ satisfying the $(r,\delta)$ locality. The dimension computation is straightforward. Our construction does not need last $l$ rows added in \cite[Construction 1]{guruswami2026quantumhierarchicallocallyrecoverable} via random sampling to ensure the dual containment requirement of CSS codes. The additional factor of $c$ comes from the EAQECC construction. The $(r,\delta)$ locality of $Q$ then follows from 
\Cref{cor:EA-Qc}, since the local restricted code 
$\textup{EACSS}(C_X|_{J_i}, C_Z|_{J_i})$ has minimum distance at 
least $\delta$ by the Vandermonde structure of $G_{a+1}$ and 
$(H_{a+1})_{\delta-1}$.

\end{proof}

\begin{remark}
For fixed $r$ and $\delta$, the first $m$ rows of the matrices $H_X$ and $H_Z$ in \Cref{cons:construction2}
are sparse. We add randomly sampled $\ell$ rows which may not necessarily be sparse. These rows are added to boost the code distance. However, the sparsity of matrices $H_X$ and $H_Z$ (up to a threshold) can be controlled by controlling the weight of the additional $\ell$ rows. Sparsity is a feature that is useful in interpreting this family as an
asymptotic entanglement-assisted quantum LDPC-type family. Hence, \Cref{prop:random_eaqlrc} is significant not only because it proves that \Cref{cons:construction2}
indeed produces EAQLRCs, but also because it is a step towards building an asymptotic family of entanglement-assisted quantum LDPC-type codes. 
Indeed, for fixed \(r\) and \(\delta\), the local rows of \(H_X\) and \(H_Z\) in \Cref{cons:construction2}  have row weight at most \(r+\delta-1\), and each coordinate appears in at most \(\delta-1\) local checks. The additional \(\ell\) rows may be dense unless one explicitly imposes a sparsity constraint when sampling them. When the additional \(\ell\) rows are sampled subject to a bounded-weight constraint, the resulting family has bounded stabilizer weights and bounded qudit degree.

Consequently, as
\[
n = m(r+\delta-1)\to\infty,
\]
the resulting family has bounded stabilizer weights and bounded qudit degree.
In particular, \Cref{prop:random_eaqlrc} shows that one can obtain an infinite family of
EAQLRCs with explicit local recovery sets and sparse stabilizer structure,
which is attractive both from the coding-theoretic perspective and for possible
quantum implementations.

\end{remark}

We conclude this section with a toy example demonstrating \Cref{prop:random_eaqlrc}. 

\begin{example}
Set $n=12$, $r=2$, $\delta=2$, 
and let $q\ge 12$. Then
\[
r+\delta-1=3,\qquad m=\frac{n}{r+\delta-1}=\frac{12}{3}=4.
\]
By \Cref{prop:random_eaqlrc}, \Cref{cons:construction2} yields an EAQLRC of dimension
\[
k = n-2m(\delta-1)+2\ell+c = 12-2\cdot 4\cdot 1 + 2\ell + c = 4+2\ell+c.
\]
Moreover, the coordinates are partitioned into the four recovery blocks
\[
\{1,2,3\},\qquad \{4,5,6\},\qquad \{7,8,9\},\qquad \{10,11,12\}.
\]
Thus each coordinate is recovered from the block containing it. For instance,
coordinates $1,2,3$ all have recovery set $\{1,2,3\}$, coordinates $4,5,6$
all have recovery set $\{4,5,6\}$, and so on. This illustrates that the local
recovery structure is completely explicit: the code consists of
small local groups of size $r+\delta-1$, each of which supports local recovery.
\end{example}

\section{Constructions of EAQLRCs from curves}
\label{sec:constructions}

In this section, we construct EAQLRCs from several families of algebraic-geometry codes over finite fields. We begin with Tamo–Barg codes on the projective line, which provide a concrete source of classical $(r,\delta)$ LRCs, and then apply the entanglement-assisted framework developed in \Cref{sec:EALRC} to obtain quantum codes with the same local recovery structure. We next present a more general fiber product construction, in which locality and availability arise from compatible morphisms between curves, and show how this framework yields both unfolded and folded EAQLRCs. Finally, we specialize this construction to Hermitian and Suzuki families, thereby obtaining additional explicit examples of EAQLRCs with availability. For each EAQLRC, the locality will be given, and the other code parameters follow from \Cref{rem:EACSS_params}.

\subsection{Tamo–Barg constructions on the projective line}

We first consider Tamo–Barg codes on the projective line, which provide a concrete source of classical $(r,\delta)$ LRCs for the entanglement-assisted constructions developed below.

\subsubsection{Classical Tamo–Barg $(r,\delta)$ LRCs}

We begin by recalling the classical Tamo–Barg construction and the corresponding dimension and distance formulas, which will serve as the input for the quantum constructions in the remainder of this subsection.

\begin{construction}
\label{cons:construction5.1}
    Let $q$ be a prime power. Choose positive integers $r,\delta \ge 2$ such that $(r+\delta-1) \mid (q-1)$. Let $\ell \in [q]_0$. Define the sets
\[
S = \left\{ i \in [\ell] \;:\; i \not\equiv r, r+1, \dots, r+\delta-2 \pmod{r+\delta-1} \right\},
\]

\[
T = \mathbb{F}_q^*.
\]
The $(r,\delta)$ Tamo--Barg code is defined as \({C}_{\mathrm{TB}} = ev_T\big( \mathbb{F}_q[X]_S \big)
\) is an $(r,\delta)$ LRC of length $n$, dimension $k = |S|= r \left\lfloor \frac{\ell}{r+\delta-1} \right\rfloor
+
\min\!\left( \ell \bmod (r+\delta-1),\, r \right)$, and
minimum distance
\[
d = (q-1) - k + 1 - \left(\left\lceil \frac{k}{r} \right\rceil - 1\right)(\delta - 1).
\]     
\end{construction}

 Note that the evaluation map is injective on $\mathbb{F}_q[X]_S$,  and we can 
 determine $|S|$ by partitioning the set of exponents into blocks of length
$r+\delta-1$. Each complete block contains exactly $r$ admissible exponents, while the remaining partial block contributes
\[
\min \left\{ \ell \bmod (r+\delta-1),\, r \right\}
\]
additional exponents.

\subsubsection{EAQLRCs from Tamo–Barg codes}

Having introduced the classical Tamo–Barg family, we now apply the entanglement-assisted construction from \Cref{sec:EALRC} to obtain EAQLRCs with the same locality parameters.

\begin{theorem}[$(r, \delta)$ EAQLRCs from Tamo--Barg classical LRCs]
\label{thm:EA_QTB}
Let $C \subseteq \mathbb{F}_q^n$ be a classical Tamo--Barg $(r, \delta)$ LRC, with $r \geq \delta -1$, $n = q-1 = m(r+\delta - 1)$, $k=\dim(C)$, $d(C) = n - k + 1 - \Big(\Big\lceil \frac{k}{r} \Big\rceil - 1\Big)(\delta - 1)$ and locality $r$. Let $Q:= \textup{EACSS}(C,C)$. Assume that the local EACSS recovery channels
associated with the recovery sets
satisfy the conditions of
\Cref{cor:EA-Qc}. Then $Q$ is an $(r,\delta)$ EAQLRC with parameters as in \Cref{rem:EACSS_params}.
\end{theorem}

\begin{proof} The Tamo--Barg code $C$ is an optimal $(r,\delta)$ LRC with recovery 
sets of size exactly $r+\delta-1$. The result follows from 
\Cref{cor_optimal} and the assumption that the recovery sets satisfy the conditions of \Cref{cor:EA-Qc}.
\end{proof}

Next, we illustrate \Cref{thm:EA_QTB} with a small example. 

\begin{example} \label{ex:EAQLRC_TB}
Let
\(q=13\), \(r=2\), and \( \delta=2\).
Then $r+\delta-1=3 \mid (q-1)=12$,
so the Tamo--Barg construction applies. Take the code $C$ obtained by
evaluating on $\mathbb{F}_{13}^\ast$, so that
\(
n=q-1=12.
\)
Choose $\ell=5$, and let
\[
S=\{\, i\in [5]_0 : i \not\equiv 2 \pmod 3 \,\}
   = \{0,1,3,4\}.
\]
Hence
\(
k = |S| = 4
\).
By \Cref{cons:construction5.1}, the classical Tamo--Barg code $C$ is an
$(r,\delta)=(2,2)$ LRC of length $12$ and dimension $4$.
Its minimum distance is
\[
d(C)
=
n-k+1-\left(\left\lceil \frac{k}{r}\right\rceil-1\right)(\delta-1)
=
12-4+1-(2-1)(1)
=
8.
\]
Let $Q:= \textup{EACSS}(C,C)$. Assume that the local EACSS recovery channels
satisfy the conditions of
\Cref{cor:EA-Qc}. Then, by \Cref{thm:EA_QTB}, since $r = 2 \geq \delta-1=2-1 = 1$ the entanglement-assisted code
\(
Q = \textup{EACSS}(C,C)
\)
is a $(2,2)$ EAQLRC.
Since $Q$ is constructed from the pair of codes $(C,C)$, the entanglement parameter is
\(
c=\dim(C)-\dim(\mathrm{Hull}(C)).
\)
A direct computation gives
\(
\dim(\mathrm{Hull}(C))=3,
\)
so
\(
c=4-3=1.
\)

One can also compute the distance bound 
explicitly as follows using software such as \cite{BallEtAl2021CodingTheoryMacaulay2, BosmaCannonPlayoust1997Magma, Macaulay2}. Since
\[
\dim(C^\perp)=12-4=8>\dim(C)=4,
\]
we have $C^\perp\not\subseteq C$, so the relevant case is
\(
d_{EA}\ge \mathrm{wt}\bigl(C^\perp\setminus \mathrm{Hull}(C)\bigr)
\).
A direct computation shows that $\mathrm{wt}(C^\perp)=3$; in fact, there is
a weight-$3$ codeword of $C^\perp$ supported on the coordinates
\[
\{1,3,9\}
\]
(for example, with nonzero entries $3,9,1$ on those positions) which does not
belong to $\mathrm{Hull}(C)$. Hence
\(
\mathrm{wt}\bigl(C^\perp\setminus \mathrm{Hull}(C)\bigr)=3\),
and therefore
\(
d_{EA}\ge 3.
\)

To make the locality explicit, index the coordinates of $C$ by the evaluation
points in $\mathbb{F}_{13}^\ast=\{1,2,\dots,12\}$. Since
\(
r+\delta-1=3,
\)
the relevant recovery sets are the cosets of the subgroup
\(
H=\{1,3,9\}\subseteq \mathbb{F}_{13}^\ast.
\)
These cosets are
\(
\{1,3,9\} \), \(\{2,5,6\}\), \(\{4,10,12\}\), \(\{7,8,11\}\).
Hence, recovery sets for each coordinate are:
\[
R_1=R_3=R_9=\{1,3,9\},
\]
\[
R_2=R_5=R_6=\{2,5,6\},
\]
\[
R_4=R_{10}=R_{12}=\{4,10,12\},
\]
and
\[
R_7=R_8=R_{11}=\{7,8,11\}.
\]
Thus, every coordinate belongs to a local group of size $3$. Since
\(
(r,\delta)=(2,2),
\)
any single erased qudit can be recovered by accessing the other two qudits in
its recovery set. In particular, the EAQLRC $Q=\textup{EACSS}(C,C)$ has locality
$r=2$ and local erasure tolerance $\delta-1=1$, inherited from the underlying
classical Tamo--Barg code via \Cref{thm:EA_QTB}. Thus this example gives an explicit $(2,2)$ EAQLRC with parameters
\(
[[12,\;5,\;\ge 3;\;1]].
\)
\end{example}

\begin{example}\label{ex:EAQLRC_TB}
Let $q=13$, $r=2$, and $\delta=3$. Then $r+\delta-1=4\mid(q-1)=12$, so the Tamo--Barg construction applies. Take the code $C$ obtained by evaluating on $\mathbb F_{13}^*$, so that $n=q-1=12$. Choose $\ell=10$, and let
\[
S=\{i\in[10]_0:i\equiv0,1\pmod4\}=\{0,1,4,5,8,9\}.
\]
Hence $k=|S|=6$. By \Cref{cons:construction5.1}, the classical Tamo--Barg code
\[
C=\operatorname{ev}_{\mathbb F_{13}^*}(\mathbb F_{13}[X]_S)
\]
is an optimal $(2,3)$ LRC of length $12$ and dimension $6$. Its minimum distance is
\[
d(C)=12-6+1-\left(\left\lceil\frac62\right\rceil-1\right)(3-1)=3.
\]

Let $Q=\operatorname{EACSS}(C,C)$. To compute its parameters, let $\alpha=2$ be a primitive element of $\mathbb F_{13}^*$ and consider the subgroup
\[
H=\langle\alpha^3\rangle=\{1,8,12,5\}.
\]
The recovery sets are the multiplicative cosets of $H$, given explicitly by
\[
R_1=R_5=R_8=R_{12}=\{1,5,8,12\},
\]
\[
R_2=R_3=R_{10}=R_{11}=\{2,3,10,11\},
\]
and
\[
R_4=R_6=R_7=R_9=\{4,6,7,9\}.
\]
Each recovery set has size $4=r+\delta-1$, and every subset of at most $\delta-1=2$ erased coordinates within a recovery set is recoverable using the remaining coordinates and the receiver's ebits.
Every polynomial defining $C$ has the form $f(X)=A(X^4)+XB(X^4)$, where $\deg A,\deg B\le2$. Since $X^4$ is constant on each coset and takes three distinct values, interpolation shows that
\[
C=D_1\oplus D_2\oplus D_3,
\]
where
\[
D_j=\{(a+bx)_{x\in\Delta_j}:a,b\in\mathbb F_{13}\}.
\]
Each $D_j$ is a $[4,2,3]_{13}$ MDS code. For each $\Delta_j=\alpha^{j-1}H$, a generator matrix of $D_j$ is
\[
G_j=
\begin{pmatrix}
1&1&1&1\\
\alpha^{j-1}&8\alpha^{j-1}&12\alpha^{j-1}&5\alpha^{j-1}
\end{pmatrix}.
\]
Since
\[
G_jG_j^T=
\begin{pmatrix}4&0\\0&0\end{pmatrix},
\]
we have $\dim\Hull(D_j)=1$. Consequently,
\[
\dim\Hull(C)=3,\qquad c=6-3=3.
\]
Thus the entanglement-assisted quantum dimension is
\[
k_Q=n-2k+c=12-2(6)+3=3.
\]

By MDS duality, each $D_j^\perp$ is a $[4,2,3]_{13}$ MDS code. Moreover, $D_j^\perp\setminus\Hull(D_j)$ contains a word of weight $3$. Hence
\[
\mathrm{wt}(C^\perp\setminus\Hull(C))=3,
\]
and by \Cref{rem:EACSS_params}, we have that $d_{EA}\ge3.$

We now verify the quantum locality. One can see that $C=D_1\oplus D_2\oplus D_3$, that is, the EACSS construction decomposes into three independent local codes $Q_j=\operatorname{EACSS}(D_j,D_j)$, each with parameters $[[4,1,\ge3;1]]_{13}$. Let $E_B=\bigcup_{j=1}^3E_{B,j}$ denote the receiver's ebits, where $|E_{B,j}|=1$. For every $j\in[3]$ and every $R\subseteq\Delta_j$ with $|R|\le\delta-1=2$, the local distance bound guarantees a recovery channel
\[
\mathrm{Rec}_{R,\Delta_j}:M\!\left((\Delta_j\setminus R)\cup E_B\right)\longrightarrow M\!\left(\Delta_j\cup E_B\right).
\]
Since the augmented codespace decomposes as $\widetilde{\mathcal Q}=\bigotimes_{j=1}^3\widetilde{\mathcal Q}_j$, these channels satisfy
\[
(\mathrm{Rec}_{R,\Delta_j}\otimes\mathrm{id}_{[12]\setminus\Delta_j})
\left(\widetilde\rho_{[12]\setminus R,E_B}\right)
=\widetilde\rho_{[12],E_B}
\]
for every joint encoded state $\widetilde\rho$. Thus the compatibility condition of \Cref{cor:EA-Qc} holds. Since $|\Delta_j|=4=r+\delta-1$, the code $Q$ satisfies \Cref{def:r-delta-QLRC}. Therefore, $Q$ is a $(2,3)$ EAQLRC with parameters $[[12,3,\ge3;3]]_{13}$.
\end{example}

\subsubsection{Folded Tamo–Barg EAQLRCs}

We next consider folds of these codes, which preserve the local recovery structure while producing additional EAQLRC families over larger alphabets.  Notice that $s=1$ recovers \Cref{thm:EA_QTB}.

\begin{theorem}[$(r,\delta)$ entanglement-assisted folded quantum Tamo--Barg codes]
\label{thm:EA_FQTB}
Consider  $C = \mathrm{ev}(\mathbb{F}_q[X]_S) \subseteq \mathbb{F}_q^{q-1}$, the classical $(r,\delta)$ Tamo--Barg code evaluated over $\mathbb{F}_q^*$. 
Let $s \mid \frac{q-1}{r+\delta - 1}$ and let $C^{(s)}$ denote the $s$-fold of $C$, so that $C^{(s)} \subseteq (\mathbb{F}_q^s)^{(q-1)/s}.$ 
Then there exists an EACSS code \( \widetilde Q=\operatorname{EACSS}\left(C^{(s)},C^{(s)}\right) \) whose parameters follow from \Cref{rem:EACSS_params}. Assume that the local EACSS recovery channels
associated with the recovery sets
conditions of
\Cref{cor:EA-Qc}. Then, \(\widetilde Q\) is an EAQLRC with \((r,\delta)\) locality.
\end{theorem}
\begin{proof} It is well known that the classical Tamo--Barg code $C$  is a classical LRC with locality \(r\). Furthermore, as shown in \cite[Lemma 57 and Lemma 58]{golowich2025quantum}, the folded code ${C^{(s)}} \subseteq (\mathbb{F}_{q}^s)^{(q-1)/s}$ is an LRC and inherits the locality $r$ from $C$. Now consider the entanglement-assisted quantum code $\widetilde{{Q}} =\textup{EACSS}({C^{(s)}}, {C^{(s)}}).$ The parameters of $\tilde{Q}$ follow directly from the entanglement-assisted construction. Finally, by \Cref{cor:EA-Qc} the resulting entanglement-assisted quantum code ${\tilde{Q}}$ inherits the locality $r$ of ${C^{(s)}}$.
\end{proof}

\begin{remark} By using the entanglement-assisted construction, one can avoid having $C^{\perp} \subseteq C,$ eliminating the need to verify the dual-containment condition required in the CSS setting. 
In addition, first folding and then applying the entanglement-assisted construction yields the same result as first forming the entanglement-assisted code and then folding.
    
\end{remark}

We close this subsection with an example demonstrating \Cref{thm:EA_FQTB}.

\begin{example}[Entanglement-assisted folded quantum Tamo--Barg code]\label{ex:EA_FQTB}
Let $q=17$, $r=2$, and $\delta=3$. Let $\mathbb F_{17^2}=\mathbb F_{17}(\theta)$, where $\theta^2=3$. Since $\frac{q-1}{r+\delta-1}=\frac{16}{4}=4$, we choose the folding parameter $s=2$, as $2\mid4$. Consider the classical Tamo--Barg code
\[
C=\operatorname{ev}_{\mathbb F_{17}^*}(\mathbb F_{17}[X]_S),\qquad S=\{0,1,4,5,8,9,12,13\}.
\]
By \Cref{cons:construction5.1}, $C$ is an optimal $(2,3)$ LRC with parameters $[16,8,3]_{17}$. Let $\alpha=3$ be a primitive element of $\mathbb F_{17}^*$ and $H=\langle\alpha^4\rangle=\{1,13,16,4\}$. Order the evaluation points as
\[
1,3,\ 13,5,\ 16,14,\ 4,12,\ 9,10,\ 15,11,\ 8,7,\ 2,6.
\]
With this ordering, the $2$-fold code is
\[
C^{(2)}=\left\{\bigl((c_1,c_2),(c_3,c_4),\ldots,(c_{15},c_{16})\bigr):(c_1,\ldots,c_{16})\in C\right\}\subseteq(\mathbb F_{17}^2)^8.
\]

We identify $\mathbb F_{17}^2$ with $\mathbb F_{17^2}$ via $(a,b)\mapsto a+\theta b$. We will verify below that the resulting folded code is $\mathbb F_{17^2}$-linear. All duals and hulls are taken with respect to the Euclidean inner product over $\mathbb F_{17^2}$.

Every polynomial defining $C$ has the form \( f(X)=A(X^4)+XB(X^4)\), where \(\deg(A),\deg(B)\leq 3\). The chosen ordering pairs $h$ with $3h$ in the first four folded coordinates and $3^2h$ with $3^3h$ in the remaining four, for $h\in H$. For $h\in H$, we have $h^4=1$ and \( (3h)^4=3^4h^4=13\).  Therefore, \( f(h)=A(1)+hB(1) \) and \( f(3h)=A(13)+3hB(13)\).  Under the identification \( (x,y)\longmapsto x+\theta y\),  the folded symbol obtained from the pair $(h,3h)$ is \[ \begin{aligned} f(h)+\theta f(3h) &= A(1)+hB(1) +\theta\bigl(A(13)+3hB(13)\bigr)\\ &= \bigl(A(1)+\theta A(13)\bigr) +h\bigl(B(1)+3\theta B(13)\bigr). \end{aligned} \] Since $A$ and $B$, each of degree at most $3$, can be prescribed independently at the four distinct values of $X^4$ corresponding to the four cosets of $H$, the first four folded coordinates form the code \[ D_1 = \left\{ (a+bh)_{h\in H} : a,b\in\F_{17^2} \right\}. \] The same calculation applied to the remaining four folded coordinates gives an isomorphic code \[ D_2 = \left\{ (a+bh)_{h\in H} : a,b\in\F_{17^2} \right\}. \] Consequently, \( C^{(2)}=D_1\oplus D_2\). 
We verify that $C^{(2)}$ is $\mathbb F_{17^2}$-linear. Since $3$ is a quadratic nonresidue in $\mathbb F_{17}$, the polynomial $X^2-3$ is irreducible, and $\{1,\theta\}$ is an $\mathbb F_{17}$-basis of $\mathbb F_{17^2}$. Moreover, the four values $\alpha^{4j}$, $j=0,1,2,3$, are distinct. Thus, by polynomial interpolation, the values $A_j=A(\alpha^{4j})$ and $B_j=\alpha^jB(\alpha^{4j})$ can be chosen independently in $\mathbb F_{17}$. It follows that
\[
a_1=A_0+\theta A_1,\quad b_1=B_0+\theta B_1,\quad
a_2=A_2+\theta A_3,\quad b_2=B_2+\theta B_3
\]
range independently over $\mathbb F_{17^2}$. Consequently,
\[
C^{(2)}=\left\{\bigl((a_1+b_1h)_{h\in H},(a_2+b_2h)_{h\in H}\bigr):a_1,b_1,a_2,b_2\in\mathbb F_{17^2}\right\}.
\]
This description is closed under addition and scalar multiplication by every $\lambda\in\mathbb F_{17^2}$, proving that $C^{(2)}$ is $\mathbb F_{17^2}$-linear. Furthermore, folding preserves the $\mathbb F_{17}$-dimension $8$, so $\dim_{\mathbb F_{17^2}}C^{(2)}=8/2=4$, consistent with the explicit decomposition $C^{(2)}=D_1\oplus D_2$.

Each $D_j$ is a $[4,2,3]_{17^2}$ Reed--Solomon code. Consequently, $C^{(2)}$ has parameters $[8,4,3]_{17^2}$.
 A generator matrix for each $D_j$ is
\[
G_j=\begin{pmatrix}
1&1&1&1\\
1&13&16&4
\end{pmatrix}.
\]
Since $G_jG_j^T=\operatorname{diag}(4,0)$, we have $\dim\operatorname{Hull}(D_j)=1$. Consequently, $\dim\operatorname{Hull}(C^{(2)})=2$ and $c=4-2=2$. Therefore, the code $Q=\operatorname{EACSS}(C^{(2)},C^{(2)})$ has quantum dimension
\[
k_Q=8-2(4)+2=2.
\]
Moreover, each local code $Q_j=\operatorname{EACSS}(D_j,D_j)$ has entanglement parameter $c_j=1$ and quantum dimension $4-2(2)+1=1$. By MDS duality, $d((C^{(2)})^\perp)=3$. Moreover,
\[
v=(0,3,2,12,0,0,0,0)\in(C^{(2)})^\perp\setminus\operatorname{Hull}(C^{(2)})
\]
has weight $3$. Hence, by \Cref{rem:EACSS_params},
\[
d_{EA}\ge\operatorname{wt}\left((C^{(2)})^\perp\setminus\operatorname{Hull}(C^{(2)})\right)=3.
\]

The folded recovery sets are $R_1=\{1,2,3,4\}$ and $R_2=\{5,6,7,8\}$. One can see that $C^{(2)}=D_1\oplus D_2$, that is, the EACSS construction decomposes into two local codes $Q_j=\operatorname{EACSS}(D_j,D_j)$, $j=1,2$, each with parameters $[[4,1,\geq 3;1]]_{17^2}$. Let $E_B=E_{B,1}\cup E_{B,2}$ denote the receiver's ebits, where $|E_{B,j}|=1$. For every $T\subseteq R_j$ with $|T|\leq\delta-1=2$, the local distance bound guarantees a recovery channel
\[
\mathrm{Rec}_{T,j}:M\!\left((R_j\setminus T)\cup E_B\right)\longrightarrow M\!\left(R_j\cup E_B\right).
\]
Indeed, each local code corrects any two erased qudits, and the direct-sum decomposition allows its recovery channel to act as the identity on $E_{B,3-j}$. Since the augmented global codespace decomposes blockwise, for every code state $\rho\in Q$, with corresponding joint state $\widetilde{\rho}$, we have
\[
(\mathrm{Rec}_{T,j}\otimes\mathrm{id}_{[8]\setminus R_j})
\left(\widetilde{\rho}_{[8]\setminus T,\,E_B}\right)
=\widetilde{\rho}_{[8],E_B}.
\]
Thus, the local recovery channels satisfy the compatibility condition of \Cref{cor:EA-Qc}. Since $|R_j|=4=r+\delta-1$, the code $Q$ is a $(2,3)$ EAQLRC.

Therefore, the construction yields an $(2,3)$ Foldeded Tamo--Barg code in the sense of \Cref{thm:EA_FQTB} with parameters $[[8,2,\ge3;2]]_{17^2}$ and explicit quantum recovery.
\end{example}

\subsection{A general fiber product framework with availability}

While the Tamo–Barg family provides a concrete first example, a broader geometric source of locality and availability is furnished by fiber products of curves.
In this more general algebraic-geometric framework, locality and availability arise from the geometry of fiber products and the associated projection maps, which we now describe.

\subsubsection{Classical fiber product codes}

We begin by recalling the general fiber product setup and the hypotheses under which the fibers of the projection maps give rise to recovery sets with availability.

\begin{construction}
Let $X$,$Y$,$Y_1,\dots,Y_t$
be smooth projective absolutely irreducible algebraic curves over $\F_q$,
defined together with regular surjective separable maps arranged in the
commutative diagram
\[\begin{tikzcd}
& X \arrow[dl, description, "g_1"'] 
     \arrow[dd, description, "g", pos=0.25] 
     \arrow[dr, description, "g_t"] & \\
Y_1 \arrow[dr, description, "h_1"'] 
    & \cdots 
    & Y_t \arrow[dl, description, "h_t"] \\
& Y &
\end{tikzcd}
\]
where $g: X \to Y$ is a map of degree $d_g$ and for each $j\in [t]$, there exist separable morphisms
\[
g_j : X \to Y_j, \qquad h_j : Y_j \to Y, \qquad j \in [t],
\]
such that the diagram is commutative, that is,
\[
h_1 \circ g_1 = h_2 \circ g_2 = \cdots = h_t \circ g_t,
\] where $\deg(g_j) = r + \delta -1$. This implies that $\deg(g) = d_g =\deg(h_j)(r+\delta -1)$ for all $j \in [t]$.
Moreover, this construction identifies $X$ with the fiber product
\[
X\cong Y_1\times_Y \cdots \times_Y Y_t,
\]
which means that
\[
g^*(\F_q(Y))=\bigcap_{j=1}^t g_j^*(\F_q(Y_j))
\subseteq \F_q(X).
\]
Here, \(g^*\) denotes the pullback of rational functions induced by the morphism \(g\), given by \(g^*(f)=f\circ g\), identifying \(\mathbb{F}_q(Y)\) with a subfield of \(\mathbb{F}_q(X)\).
We assume that the maps $g,g_1,\dots,g_t$ satisfy the following hypotheses, which ensure that the fibers of the maps $g_j$ yield $t$ recovery sets and that interpolation on each fiber is possible:
\begin{enumerate}
\item[(i)] Suppose that for each $j\in [t]$, $\F_q(X)=g_j^*(\F_q(Y_j))(x_j),$
and $\F_q(X)=g^*(\F_q(Y))(x_1,\dots,x_t),$
where $x_j$ is a primitive element of the corresponding separable extension. We may also write $\F_q(X)=\F_q(Y)(x),$
where $x$ is a primitive element, and denote its degree by $h$.

\item[(ii)] Let $B\subseteq A\subseteq X(\F_q)$, where $A$ is a set of the form
\[
A=g^{-1}(S)
\]
for some set $S\subseteq Y(\F_q)$. Assume that the subset $B$ can be partitioned
into pairwise disjoint subsets in $t$ different ways:
\[
B= \bigcup_{y_1 \in g_1(B)} g^{-1}_1(y_1) = \bigcup_{y_2 \in g_2(B)} g^{-1}_2(y_2) = \dotsm = \bigcup_{y_t \in g_t(B)} g^{-1}_t(y_t) \] Therefore for each fixed $j\in [t]$, \[ B=\bigsqcup_{y_j\in g_j(B)} g_j^{-1}(y_j)
\]
so that all fibers of $g_j$ over $g_j(B)\subseteq Y_j(\F_q)$ consist of
$\F_q$-rational points of $X$.

\item[(iii)] For every $j\in [t]$ and every $y_j\in g_j(B)$,
\[
|g_j^{-1}(y_j)|=r+\delta-1.
\]

\item[(iv)] Finally, assume that the fibers 
$(y_1,\dots,y_t)\in g_1(B)\times\cdots\times g_t(B)$ satisfy the following:
\[
\left|\bigcap_{j=1}^t g_j^{-1}(y_j)\right|\le \delta - 1.
\]
\item[(v)] For each $j\in [t]$, let $x_j\in \F_q(X)$ satisfy the following conditions:
\begin{enumerate}
    \item for every $j\in [t]$ and every fiber $R_j=g_j^{-1}(y_j)$, the restriction of $x_i$ to $R_j$ is constant for all $i\neq j$;
    \item for every $j\in [t]$ and every fiber $R_j=g_j^{-1}(y_j)$, the values of $x_j$ are pairwise distinct on $R_j$;
    \item there exists a rational point $P_\infty\in X(\F_q)$ and an integer $h\ge 0$ such that
    \[
    (x_j)_\infty\le hP_\infty,\qquad j\in [t].
    \]
\end{enumerate}
\end{enumerate}
 
Let $D$ be a positive divisor on $Y$ of degree $\ell$ such that $D = \ell Q_{\infty},Q_{\infty} \in Y(\F_q)$, $Q_{\infty} \notin g(B)$ and let $\{f_1,\dots,f_m\}$ be a basis of the linear space
\[
\mathcal L(D)\subseteq \F_q(Y),
\] that is, the functions $f_i, i =1,\ldots, m$ are contained in $\mathbb{F}_q(Y)$ and thus $f_i \circ g$ is constant on fibers of $g$. By the Riemann-Roch theorem, $m \geq \ell - g_Y + 1$ where $g_Y$ is the genus of $Y$. Consider the following polynomial space
\[
V:=\operatorname{Span}_{\F_q}
\left\{
x_1^{i_1}\cdots x_t^{i_t}f_k
:\;
i_j\in [r]_0,\ k \in [m]
\right\}
\subseteq \F_q(X).
\]
By primitive element assumptions and separability of extensions, the set of functions $\{x_1^{i_1}\cdots x_t^{i_t}f_k\}$ are linearly independent and form a basis of $V$, thus $\dim(V) = r^tm$. 
The fiber product Tamo--Barg-like code of length $n=|B|$ is constructed as the
image of the evaluation map
\[
\ev_B:V\longrightarrow \F_q^{|B|},
\qquad
F\longmapsto (F(P))_{P\in B}.
\]

The role of the fiber product hypotheses is to ensure that each recovery set is controlled by one projection while remaining largely independent of the others. Condition~(a) guarantees that, along a fiber associated with one recovery map, the functions coming from the other maps restrict to constants on the fiber, while condition~(b) ensures that enough variation remains within the chosen fiber to support local interpolation and recovery. In this way, the geometry of the fiber product translates into the combinatorial structure required for locality and availability.
   \end{construction}

Under these hypotheses, the corresponding evaluation code on the fiber product yields a classical LRC with availability, as stated in the following theorem.

The following proposition follows directly from \cite[Theorem 3.2]{Barg_2015} and the discussion preceding it, and can be regarded as an extension of \cite[Theorem 5.1]{Barg_2015}.

\begin{theorem} [$(r,\delta)$ fiber product code with availability $t$]
\label{codeTamolike}
The subspace $C(D,g_1,\dots,g_t)\subseteq \F_q^{|B|}$ forms an
$(n,k,(r,\delta),t)$ linear LRC with availability $t$ and
parameters
\[
n=|B|,
\]
\[
k=r^t m\ge r^t(\ell-g_Y+1),
\]
\[
d\ge n-\ell \deg(g) -t(r-1)h,
\]
provided that the right-hand side of the inequality for $d$ is a positive integer and that the evaluation map is injective.
\end{theorem}

As in the case of Tamo–Barg codes, the fiber product construction also admits a folded version, which we now record for later use in the entanglement-assisted setting. 

\begin{definition}\label{foldedTamo}
Given a folding parameter $s \mid \frac{|B|}{r+\delta-1},$
we define the \emph{folded fiber product code}  $C^{(s)}(D,g_1,\ldots,g_t)$ over alphabet $\F_{q}^s$ by grouping  $s$ consecutive coordinates. Assume that
\[
B=\bigsqcup_{y\in g_j(B)} g_j^{-1}(y)
\]
for some fixed $j\in [t]$, and that each fiber $g_j^{-1}(y)$ has size $r+\delta-1$. Let \[P_1,\ldots, P_{|B|}
\] be an ordering of $B$ that is
compatible with the partition into fibers of $g_j$, in the sense that each block of $r+\delta-1$ consecutive points $P$ corresponds to one fiber. For each $i\in \left[\frac{|B|}{s}\right],$
define
\[
\widetilde B_i:=\{P_{(i-1)s+1},\dots,P_{is}\}.
\]

Now for $F\in V$, define the folded evaluation map
\[
\operatorname{ev}_B^{(s)}(F)
=
\Big(
(F(P))_{P\in \widetilde B_1},
\ldots,
(F(P))_{P\in \widetilde B_{|B|/s}}
\Big)\in (\F_q^s)^{|B|/s}.
\]

The image of $V$ under this map is called the folded fiber product code
\[
C^{(s)}(D,g_1,\ldots,g_t).
\]
\end{definition}

\begin{proposition} The code $C^{(s)}(D,g_1,\ldots,g_t)$ in \Cref{foldedTamo} has parameters
\[
n = \frac{|B|}{s}, k = r^tm/s,
\]
and distance $d \geq \left \lceil(|B|-\ell \deg(g) -t(r-1)h)/s\right\rceil$
obtained by grouping \(s\) consecutive coordinates at a time in an ordering compatible with the fibers of \(g_j\).
\end{proposition}

\begin{proof}
  The length and dimension follow directly from \Cref{codeTamolike} and the definition of folding, while the stated distance estimate is a lower bound inherited from the corresponding unfolded code: grouping coordinates into blocks of size $s$ can only decrease the Hamming weight by a factor of at most $s$, so the minimum distance of the folded code is at least the minimum distance of the original code divided by $s$.
\end{proof}

\subsubsection{EAQLRCs from fiber product codes}

We now apply the entanglement-assisted framework of \Cref{sec:EALRC} to the fiber product codes constructed above, thereby obtaining EAQLRCs with availability.

\begin{proposition} [fiber product EAQLRCs]\label{prop:EAQLRC_fiber}
Let $C$ be the fiber product code in  \Cref{codeTamolike}. Let $Q:=\textup{EACSS}(C,C)$. Assume that the local EACSS recovery channels
associated with the Tamo--Barg recovery sets
satisfy the conditions of
\Cref{cor:EA-Qc}. Then $Q$ is an $(r,\delta)$ EAQLRC with availability $t$.
\end{proposition}

\begin{proof} The code $C$ is an optimal 
$(r,\delta)$ LRC with recovery sets of size exactly $r+\delta-1$ and 
local dimension exactly $r$ by the fiber product construction hypotheses. 
Since $r \geq \delta-1$, the result follows from \Cref{cor_optimal} and the parameters follow from \Cref{rem:EACSS_params}. 

\end{proof}

Next, we provide an example to illustrate \Cref{prop:EAQLRC_fiber}.

\begin{example} \label{ex:fiber_EAQLRC}
Let $q=29$, $Y=\mathbb{P}^1_x$,
$Y_1=\mathbb{P}^1_u$, and $Y_2=\mathbb{P}^1_v$.
Consider the morphisms
\[
\begin{array}{ll}
h_1:Y_1\to Y & x=u^4,
\\
h_2:Y_2\to Y & x=1-v^4
\end{array}
\] whose fiber product is the smooth projective curve
\(
X:\ u^4+v^4=1.
\)
Take the following maps: 
\[
\begin{array}{ll}
g_1:X\to Y_1 & (u,v)\mapsto u,
\\
g_2:X\to Y_2& (u,v)\mapsto v,
\\
g:X\to Y & (u,v)\mapsto x=u^4=1-v^4.
\end{array}
\]
Then
\[
\deg(g_1)=\deg(g_2)=4=r+\delta-1,
\]
which gives rise to locality
\(
r=3,\delta=2.
\)
To define the evaluation set, note that in $\mathbb{F}_{29}$, the nonzero fourth
powers are
\[
\{1,7,16,20,23,24,25\}.
\]
Among these, the values
\[
x=7,\qquad x=23
\]
have the property that both $x$ and $1-x$ are nonzero fourth powers. Set
\[
S=\{7,23\}\subseteq Y(\mathbb{F}_{29}),
\qquad
B=g^{-1}(S)\subseteq X(\mathbb{F}_{29}).
\]
For each $x\in S$, the equation $u^4=x$ has $4$ solutions and the equation
$v^4=1-x$ has $4$ solutions, so each fiber of $g$ above $x$ contributes
\(
4\cdot 4=16
\)
points. Hence,
\[
n=|B|=32.
\]

Now, take the divisor
\[
D=0
\]
on $Y=\mathbb{P}^1_x$, so that
\(
\mathcal L(D)=\operatorname{Span}_{\mathbb{F}_{29}}\{1\}\),
and
\(m=\ell(D)=1\).
With $x_1=u$ and $x_2=v$, the classical fiber product code in \Cref{codeTamolike} is
obtained from
\[
V=
\operatorname{Span}_{\mathbb{F}_{29}}
\{u^{e_1}v^{e_2}:0\le e_1,e_2\le 2\}.
\]
Therefore,
\(
k=\dim(C)=r^t m = 3^2\cdot 1=9
\).
Since $Y=\mathbb{P}^1$ has genus $0$, $\deg(g)=16$, and because $u$ and
$v$ each have pole order $1$ at infinity, we may take
\(
h=1
\).
Hence \Cref{codeTamolike} gives the classical distance bound
\[
d(C)\ge n-\ell\deg(g)-t(r-1)h
      = 32-1 \cdot 16-2\cdot 2\cdot 1
      = 12.
\]
Assume that the local EACSS recovery channels
satisfy the conditions of
\Cref{cor:EA-Qc}. Applying \Cref{prop:EAQLRC_fiber}, the entanglement-assisted quantum code
\(
Q=\textup{EACSS}(C,C)
\)
is an $(r,\delta)=(3,2)$ EAQLRC with availability
\(
t=2
\).

One may check that 
\(
\dim(\Hull(C))=5
\),
so
\(
c=\dim(C)-\dim(\Hull(C))=9-5=4
\).
Therefore, the quantum dimension is
\[
k_Q=n-2k+c = 32-2\cdot 9+4 = 18.
\]
We can also compute the entanglement-assisted distance bound explicitly. Since
\[
\dim(C^\perp)=32-9=23>\dim(C)=9,
\]
we have
\(
C^\perp\not\subseteq C
\). According to \Cref{rem:EACSS_params},
\(
d_{EA}\ge \mathrm{wt}\!\bigl(C^\perp\setminus \Hull(C)\bigr)
\).
A direct computation shows that
\(
\mathrm{wt}\!\bigl(C^\perp\setminus \Hull(C)\bigr)=4
\),
since there is a codeword of $C^\perp$ supported on the first four
coordinates with nonzero entries
\(
28,17,12,1,
\)
and this codeword does not belong to $\Hull(C)$. As a result, 
\(
d_{EA}\ge 4
\).

To make the locality and availability explicit, fix a coordinate
corresponding to a point
\(
P=(u,v)\in B
\).
The two recovery sets come from the fibers of $g_1$ and $g_2$:
\[
R_1(P)=\{(u,z)\in B : z^4=1-u^4\},
\]
\[
R_2(P)=\{(z,v)\in B : z^4=1-v^4\}.
\]
Since $u^4\in S$ and $1-u^4$ is a fourth power, each of these sets has size
\[
|R_1(P)|=|R_2(P)|=4=r+\delta-1.
\]
Moreover,
\[
R_1(P)\cap R_2(P)=\{P\},
\]
so the two recovery sets are disjoint outside of $P$. 
Because $\delta=2$, any single erased coordinate can be recovered from the
other \(
r=3
\)
coordinates in either recovery set. Thus, under the  assumption of \Cref{cor:EA-Qc}, the code $Q=\operatorname{EACSS}(C,C)$ is a $(3,2)$ EAQLRC with availability $t=2$ and parameters $[[32,18,\geq4;4]]_{29}$.
\end{example}

\subsubsection{Folded fiber product EAQLRCs}

The same argument applies to the folded family, giving folded EAQLRCs derived from the fiber product construction.

\begin{proposition}[Folded fiber product $(r,\delta)$ EAQLRC with availability $t$] \label{prop:folded_fiber_EAQLRC} Let $C^{(s)}$ be the code in \Cref{foldedTamo}. Assume that the local EACSS recovery channels
associated with the Tamo--Barg recovery sets
satisfy the conditions of
\Cref{cor:EA-Qc}. Then $Q = \textup{EACSS}(C^{(s)},C^{(s)})$ over the alphabet $\F_{q}^s$ is an $(r,\delta)$ EAQLRC with availability $t$ with parameters
\[n = |B|/s \hbox{ and } k = |B|/s - 2(r^tm)/s + c.\]

\end{proposition}

The distance and entanglement parameter follow from the computation of the hull, \Cref{cor:EA-Qc} and \Cref{rem:EACSS_params}. We end this subsection with an example demonstrating \Cref{prop:folded_fiber_EAQLRC}.

\begin{example}
Consider the fiber product
\(
X:\ u^4+v^4=1
\)
over $\mathbb{F}_{29}$ as in \Cref{ex:fiber_EAQLRC},  with the same evaluation set
\(
B=g^{-1}(\{7,23\})\subseteq X(\mathbb{F}_{29})
\)
so that
\(
|B|=32
\).
Recall from the previous example that the underlying classical code $C$ has parameters
\(
[n,k,d]\;=\;[32,9,\ge 12],
\)
and the corresponding EAQLRC from \Cref{prop:EAQLRC_fiber} has
$r=3$, $\delta=2$, and $t=2$.
Observe that $r+\delta-1=4$ and $\frac{|B|}{r+\delta-1}=\frac{32}{4}=8$ is even. As a result, 
we may take $s:=2$.

To define the fold explicitly, order the points of $B$ as follows. For each
fixed $x\in\{7,23\}$ and each fixed solution $u$ of $u^4=x$, the four
solutions of $v^4=1-x$ form one $g_1$-fiber. We pair these four points into
two consecutive blocks of size $2$. For example, for $x=7$ and $u=8$, we take
the two folded coordinates
\[
((7,8,3),(7,8,7))
\qquad\text{and}\qquad
((7,8,22),(7,8,26)).
\]
Doing this for all choices of $x$ and $u$ produces the $2$-folded code
\( C^{(2)}\subseteq (\mathbb{F}_{29}^2)^{16}\) 
of length $
n=\frac{|B|}{2}=16$.

Assume that the local EACSS channels satisfy the conditions of \Cref{cor:EA-Qc}. Applying \Cref{prop:folded_fiber_EAQLRC}, the entanglement-assisted quantum code
\(
\widetilde Q=\textup{EACSS}(C^{(2)},C^{(2)})
\)
is a folded EAQLRC. The \(F_{29}\)-dimension of the underlying folded code remains \(9\). However, according to our folded-dimension convention, \(C^{(2)}\) has folded dimension \( \frac{9}{2}\). A direct computation in the underlying \(F_{29}\)-linear representation gives \( \dim_{F_{29}}\operatorname{Hull}\left(C^{(2)}\right)=5\). Since \(\dim_{F_{29}} C^{(2)}=9\), the entanglement parameter is \( c=9-5=4\). Under our folded-dimension convention, the dimension of \(C^{(2)}\) is \(9/2\), so \[ k_{\mathrm{EA}}=16-2\cdot\frac{9}{2}+4=11. \]
 To compute the distance bound explicitly, observe 
that in the underlying \(F_{29}\)-linear representation, the ambient dimension is \(32\), and \( \dim_{F_{29}}\bigl((C^{(2)})^\perp\bigr)=32-9=23>9=\dim_{F_{29}}C^{(2)}\).  Hence \((C^{(2)})^\perp\not\subseteq C^{(2)}\).
According to \Cref{prop:folded_fiber_EAQLRC} and computation,
\[
d_{EA}\ge
\mathrm{wt}\!\bigl((C^{(2)})^\perp\setminus \Hull(C^{(2)})\bigr)=2.
\]
Indeed, there is a codeword of $(C^{(2)})^\perp$ supported on the first
two folded coordinates, namely
\[
\bigl((28,17),(12,1),(0,0),\dots,(0,0)\bigr),
\]
and this codeword does not belong to $\Hull(C^{(2)})$. Thus, \( d_{EA}\ge 2.
\)

We now make the folded locality explicit. A folded coordinate is of the form
\[
((x,u,v_1),(x,u,v_2)),
\]
where $(x,u,v_1)$ and $(x,u,v_2)$ lie in the same $g_1$-fiber. One explicit
family of folded recovery sets is obtained by varying $u$ while keeping the
pair $\{v_1,v_2\}$ fixed. For example, for the folded coordinate
\(
P_1=((7,8,3),(7,8,7))
\),
the folded recovery set is
\[
\widetilde R(P_1)
=
\{
((7,8,3),(7,8,7)),
((7,9,3),(7,9,7)),
((7,20,3),(7,20,7)),
((7,21,3),(7,21,7))
\}.
\]
Similarly, for
\(
P_2=((7,8,22),(7,8,26))
\),
the folded recovery set is
\[
\widetilde R(P_2)
=
\{
((7,8,22),(7,8,26)),
((7,9,22),(7,9,26)),
((7,20,22),(7,20,26)),
((7,21,22),(7,21,26))
\}.
\]
There are analogous folded recovery sets for the coordinates lying above
$x=23$.
Each such folded recovery set has size
\(
|\widetilde R(P)|=4=r+\delta-1.
\)
Since $\delta=2$, any single erased folded coordinate can be recovered from the other $r=3$ folded coordinates in its recovery set. Furthermore, a second folded recovery family is obtained analogously from the fibers of $g_2$, so \Cref{prop:folded_fiber_EAQLRC} gives availability $t=2$. Thus, under the  assumption of \Cref{cor:EA-Qc}, $\widetilde Q$ is a $[[16,\;11,\;\ge 2;\;4]]$ folded
EAQLRC with locality $r=3$ and availability $t=2$. 
\end{example}

\subsection{Hermitian and Suzuki specializations}

We next specialize the general algebraic-geometric framework to two important families of curves with many rational points, namely Hermitian and Suzuki curves.

\subsubsection{Hermitian-based EAQLRCs}

We first consider the Hermitian family, whose associated classical LRCs with availability give rise to entanglement-assisted quantum codes by the general construction.
The Hermitian family considered here may also be viewed as arising from the Artin–Schreier fiber product construction in \cite[Theorem 7.1]{HMM_HLRC_RM_2024}, and thus provides an explicit source of EAQLRCs with availability.

\begin{proposition} \label{prop:herm} Let $C$ be a code on the Hermitian curve defined above and in \cite[Theorem 7.1]{Haymaker2016LocallyRC}, where $C$ has locality $p-1$ (each recovery set has size $p)$ and availability $t$. Let $Q:= \textup{EACSS}(C,C)$. Assume that the local EACSS recovery channels
associated with the Tamo--Barg recovery sets
satisfy the conditions of
\Cref{cor:EA-Qc}. Then, $Q$ is an EAQLRC with availability $t$ and locality $p-1 $ (meaning $t$ recovery sets of size $p$), with parameters $[[n, p^tq-2p^tq - 2((q-tp^t+1)(p-1)^t) +2c, d; c]]$.

\end{proposition}

\begin{proof} The parameters follow from \Cref{rem:EACSS_params}. The Hermitian code $C$ in \cite[Theorem 7.1]{Haymaker2016LocallyRC} is an optimal $(p-1,2)$ LRC with recovery sets of size exactly $p$ and $r = p-1 \geq 1 = \delta - 1$. The result then follows from 
\Cref{cor_optimal} and the assumption that the channels satisfy the conditions of \Cref{cor:EA-Qc}.
   
\end{proof}

\subsubsection{Suzuki-based EAQLRCs}

We next investigate the family of Suzuki codes, which provides explicit examples in which the recovery sets need not all have the same size. 

We can now construct an $(r,2)$ EAQLRC with availability $2$ from the Suzuki curve. Recall that the Suzuki curve $\mathcal{S}_q$ for $q_0 = 2^m, q = 2^{2m+1}, m \geq 1$ is defined as $(x^q + x)x^{q_0} = y^{q} +y$. In \cite[Theorem 6.1]{Haymaker2016LocallyRC}, the authors show that there exists LRCs $C(D,B)$ with availability $2$ and locality $(q-1,q-2)$ on the Suzuki curve $S_q$, with the following parameters:
\begin{enumerate}
    \item over $\mathbb{F}_q$, with length $n = q(q-1)$
    and dimension $k = (q-1)(q-2);$
    
    \item over $\mathbb{F}_{q^4}$, with length $n = q(q-1)\bigl(q^2 + 2qq_0 + q + 1\bigr)$
    and dimension $k = (q-1)(q-2)\bigl(q^2 + 2qq_0 + q + 1\bigr).$
\end{enumerate}

In both cases, $B$ is obtained by removing the point at infinity together with an additional set of $q$ rational points from $S_q(\mathbb{F})$ and $|B| = |S_q(\mathbb{F})| - (q+1),$  where $\F = \F_q$ in this first case and~$\F = \F_{q^4}$.

We can now construct an $(r,2)$ EAQLRC with availability $2$ where the two recovery sets have sizes $q-1$ and $q-2$, respectively. We have the following proposition.

\begin{proposition} \label{prop:suzuki}
Let $C= C(D,B)$ be the code on the Suzuki curve $\mathcal{S}_q$ above with availability $2$ and with two recovery sets of sizes \(q-1\) and \(q-2\), equivalently with helper localities \(q-2\) and \(q-3\) in the convention where locality counts the number of other coordinates accessed. Set $c = \dim(C) - \dim(\mathrm{Hull}(C))$. Let $Q := \textup{EACSS}(C,C)$. Assume that the local EACSS recovery channels
associated with the Tamo--Barg recovery sets
satisfy the conditions of
\Cref{cor:EA-Qc}. Then $Q$ is an $(r,2)$ EAQLRC with availability $2$, recovery sets of sizes $q-1$ and $q-2$, 
and the following parameters: 
\begin{enumerate}
    \item over $\mathbb{F}_q$, $Q$ has length $n = q(q-1)$
    and dimension $k = q(q-1) - 2(q-1)(q-2) + c$.
    \item over $\mathbb{F}_{q^4}$, $Q$ has length $n = q(q-1)\bigl(q^2 + 2qq_0 + q + 1\bigr)$
    and dimension $k = q(q-1)\bigl(q^2 + 2qq_0 + q + 1\bigr) - 2((q-1)(q-2)\bigl(q^2 + 2qq_0 + q + 1\bigr)) + c$.
\end{enumerate}

\end{proposition}

The distance and the entanglement parameter $c$ follow from the computation of the hull and \Cref{rem:EACSS_params}. $C$ is an 
optimal $(r, 2)$ LRC with recovery sets of sizes $q-1$ and $q-2$, 
and $r = q-2 \geq 1 = \delta - 1$ since $q = 2^{2m+1} \geq 8$. 
The result follows from \Cref{cor_optimal}. According to \cite{Haymaker2016LocallyRC}, computing the distances of the codes above depends on the degrees of the generators $x_i$ of the corresponding function field extension and explicit generators are not available. Therefore, the degrees cannot be computed and the bound cannot be evaluated directly. Computing the distance in the EAQECC case also requires one to compute the hulls of these codes. Thus, we write the general distance formula given by the EAQECC construction.

\subsection{Summary of explicit EAQLRC families}

We conclude the section by summarizing the explicit EAQLRC families obtained above and their basic parameters in \Cref{tab:section5-ea-summary}.

\begin{table}[t]
\centering
\footnotesize
\renewcommand{\arraystretch}{1.2}
\setlength{\tabcolsep}{4pt}
\begin{tabular}{|p{4.1cm}|C{2.2cm}|C{4.9cm}|C{1.8cm}|C{1.9cm}|}
\hline
\textbf{Classical Code $C_1$} &
\textbf{Length} &
\textbf{Dimension} &
\textbf{Locality} &
\textbf{Availability} \\
\hline

\makecell[l]{Tamo--Barg code $C$\\ (\Cref{thm:EA_QTB})}
& $\displaystyle q-1$
& $\displaystyle (q-1) - 2k + c$
& $(r,\delta)$
& $1$
\\
\hline

\makecell[l]{folded Tamo--Barg code $C^{(s)}$\\ (\Cref{thm:EA_FQTB})}
& $\displaystyle \frac{q-1}{s}$
& $\displaystyle \frac{q-1}{s} - \frac{2k}{s} + c$
& $(r,\delta)$
& $1$
\\
\hline

\makecell[l]{fiber product code\\ (\Cref{prop:EAQLRC_fiber})}
& $\displaystyle |B|$
& $\displaystyle |B| - 2r^t m + c$
& $(r,\delta)$
& $t$
\\
\hline

\makecell[l]{folded fiber product\\ code $C^{(s)}$\\ (\Cref{prop:folded_fiber_EAQLRC})}
& $\displaystyle \frac{|B|}{s}$
& $\displaystyle \frac{|B|}{s} - \frac{2r^t m}{s} + c$
& $(r,\delta)$
& $t$
\\
\hline

\makecell[l]{Hermitian code\\ (\Cref{prop:herm})}
& $\displaystyle p^t q$
& $\displaystyle p^t q - 2(q-tp^t+1)(p-1)^t + c$
& $p-1$
& $t$
\\
\hline

\makecell[l]{Suzuki code over $\F_q$\\ (\Cref{prop:suzuki})}
& $\displaystyle q(q-1)$
& $\displaystyle q(q-1)-2(q-1)(q-2)+c$
& $q-1,\ q-2$
& $2$
\\
\hline

\makecell[l]{Suzuki code over $\F_{q^4}$\\ (\Cref{prop:suzuki})}
& $\displaystyle q(q-1)(q^2+2qq_0+q+1)$
& \makecell[l]{$\displaystyle q(q-1)(q^2+2qq_0+q+1)$\\
$\displaystyle -2(q-1)(q-2)$ \\ $\displaystyle (q^2+2qq_0+q+1)$\\$\displaystyle +c$}
& $q-1,\ q-2$
& $2$
\\
\hline
\end{tabular}
\caption{Summary of the explicit EAQLRC families $\textup{EACSS}(C_1,C_1)$ obtained in \Cref{sec:constructions} from classical codes $C_1$, including their length, dimension, locality, and availability. For the Tamo--Barg families, $k$ denotes the dimension of the underlying classical code. The minimum distance is bounded below as in \Cref{rem:EACSS_params}. For the folded families, the displayed dimensions assume that the folded classical codes are linear over a  field extension of degree $s$ and have the indicated divided dimensions.}
\label{tab:section5-ea-summary}
\end{table}

\section{Conclusions and Discussion}
\label{sec:concl}

Erasure recovery plays an important role in classical distributed storage. In both the classical and quantum settings, erasure recovery is often easier   than general error correction because the error locations are known. Furthermore, on hardware platforms such as superconducting as well as neutral-atom devices, various types of errors can be converted to erasures (errors whose locations are known). Entanglement-assisted codes were developed as a way to relax the orthogonality requirements for quantum CSS codes. These codes are known to improve tradeoffs involving the number of encoded logical qubits as well as the number of recoverable erasures.

In this work, we presented \emph{EAQLRCs with availability}. Classically, if an erasure can be recovered by more than one set of coordinates that are pairwise disjoint,  the code is said to be \emph{ locally recoverable with availability}. While it has been demonstrated in \cite{golowich2025quantum} that disjoint recovery sets and hence the notion of  availability cannot be defined for quantum CSS codes, we define availability in the entanglement-assisted setting. Here, we show that given two classical LRCs, we can construct an EAQLRC. We compute code parameters -- including dimension and distance -- which require determining the hull of the classical LRC. We also determine a Singleton-like tradeoff among length, dimension, distance, locality, availability, and entanglement consumption for EAQLRCs. With respect to code constructions, we present both random and explicit constructions. Our explicit constructions rely on variants of algebraic-geometry codes such as Tamo--Barg codes, folded Tamo--Barg codes, Hermitian codes, Suzuki codes, and fiber product codes. 

It remains to determine whether the explicit codes presented here are optimal with respect to the Singleton bound presented. This depends on computing the dimensions of the relevant hulls and the minimum weights of the complements appearing in the entanglement-assisted distance formula.  Determining these parameters will also allow a more precise comparison with the parameters of quantum CSS LRCs. Another interesting question is the nature of quantum hierarchical local recovery in the entanglement-assisted setting, which is a work in progress.

\bibliographystyle{abbrv}
\bibliography{EAQLRC}

\end{document}